\documentclass[final]{siamart251216}

\usepackage{amsmath}
\usepackage{stmaryrd}
\usepackage[utf8]{inputenc}
\usepackage[english]{babel}
\usepackage{accents}
\usepackage[T1]{fontenc}
\usepackage{graphicx}
\usepackage{wrapfig}
\usepackage{accents}
\usepackage{fancyhdr}
\usepackage{amsfonts}
\usepackage{amssymb}
\usepackage{dsfont}
\usepackage{caption}
\usepackage{subcaption}
\usepackage[justification=centering]{caption}
\usepackage{multirow}
\usepackage{scalerel}
\usepackage{datetime}
\usepackage{wasysym}
\usepackage{array}
\usepackage{listings}
\usepackage{csquotes}
\usepackage{algorithm}
\usepackage{booktabs}
\usepackage{algpseudocode}
\usepackage{enumerate}
\usepackage{hyperref}
\usepackage{cleveref}
\usepackage{amsopn}
\usepackage{xurl}

\ifpdf
  \DeclareGraphicsExtensions{.eps,.pdf,.png,.jpg}
\else
  \DeclareGraphicsExtensions{.eps}
\fi

\newsiamremark{remark}{Remark}
\newsiamremark{hypothesis}{Hypothesis}
\crefname{hypothesis}{Hypothesis}{Hypotheses}
\newsiamthm{claim}{Claim}
\newsiamremark{fact}{Fact}
\crefname{fact}{Fact}{Facts}

\headers{The M-Tensor Formalism}{R. Cloarec, S. Rodriguez, X. Kestelyn, F. Chinesta}
\author{Remi Cloarec\thanks{Arts et Métiers Institute of Technology, 151 Bd de l'Hôpital, Paris, France}}

\title{M-Tensor Formalism: A Non-iterative High Dimensional Least Squares Regression for Nonlinear Models with Scarce Data \thanks{Submitted to the editors DATE.
\funding{This work was funded by the ENSAM RTE Chair.}}}

\author{Rémi Cloarec$^\ddagger$\thanks{ENSAM RTE Chair, Laboratoire PIMM UMR 8006, Arts et Métiers Institute of Technology, CNRS, Cnam, F-75013 Paris, France
  (\email{remi.cloarec@ensam.eu}, \email{sebastian.rodriguez\_iturra@ensam.eu}, \email{francisco.chinesta@ensam.eu}).}
\and Sebastian Rodriguez\footnotemark[2]
\and Xavier Kestelyn\thanks{ENSAM RTE Chair, Univ. Lille, Arts et Métiers Institute of Technology, Centrale Lille, Junia, ULR 2697-L2EP, F-59000 Lille, France
  (\email{xavier.kestelyn@ensam.eu}).}
\and Francisco Chinesta\footnotemark[2]}

\numberwithin{equation}{section}

\hypersetup{
    colorlinks=true,
    linkcolor=black,
    filecolor=magenta,      
    urlcolor=cyan,
    citecolor =cyan,
    anchorcolor = green,
}

\newdateformat{mydateformat}{\THEDAY-\THEMONTH-\THEYEAR}

\DeclareMathOperator*{\Lotimes}{%
  \mathop{\mathchoice{\ \,\text{\raisebox{-.5mm}{\Large{|}}}\hspace{-1.6mm}\bigotimes\ \ }{|\hspace{-1.5mm}\bigotimes\ \ }{|\hspace{-1.5mm}\bigotimes\ \ }{|\hspace{-1.5mm}\bigotimes\ \ }}%
}

\DeclareMathOperator{\mP}{\mathcal{P}}
\DeclareMathOperator{\mO}{\mathcal{O}}
\DeclareMathOperator{\bx}{\mathbf{x}}
\DeclareMathOperator{\bC}{\mathbf{C}}

\DeclareMathOperator{\bpsi}{\pmb{\psi}}
\DeclareMathOperator{\bphi}{\pmb{\varphi}}
\DeclareMathOperator{\bPhi}{\mathbf{\Phi}}

\DeclareMathOperator{\dbx}{\delta\mathbf{x}}
\DeclareMathOperator{\dbC}{\delta\mathbf{C}}
\DeclareMathOperator{\dbphi}{\delta\pmb{\varphi}}
\DeclareMathOperator{\R}{\mathbb{R}}

\DeclareMathOperator*{\Id}{\mathbb{I}\text{d}}

\begin{document}

\maketitle

\begin{abstract}
We present a multilinear regression framework based on tensor algebra tailored to high-dimensional contexts where data is scarce. We exploit algebraic properties of a partial tensor product, namely the m-tensor product, to leverage structured equations with separated variables. The proposed method combines kernel properties along with tensor algebra to prevent explicit construction of the exponentially large feature space and tackle approximations up to hundreds of parameters while avoiding the fixed-point strategy. This is achieved by only ever employing the regression operator in a factorized form. We present this formalism along with different regularization techniques suited for low amount of data with a high number of parameters while preserving well-known matrix-based properties. We demonstrate complexity scaling on a general benchmark to show robustness for engineering problems and ease of implementation. 
\end{abstract}



\section{Introduction}

    Engineering problems are often made dependent on parameters. In certain cases, the number of those parameters might be reduced by identifying fewer latent variables as is done by model order reduction techniques. In certain scenarios the dependence has to be explicit. In that case even for problems with low latent dimensionality, making explicit use of those parameters becomes intractable. This issue is often found in design optimization problems for instance. Manipulating a large number of parameters often becomes numerically intractable because of a wall of complexity coined \emph{curse of dimensionality} \cite{bellman}. This curse translates the exponential growth of models (\emph{i.e.} the number of coefficients) and that of the required amount of data with respect to the number of parameters. In contexts where data is scarce and costly, that curse makes models simply inaccessible.

    Solutions proposed to circumvent this curse can rely on a modification of the grid required by the model, \emph{e.g.} with a sparse grid \cite{smolyak2013}, or with an \emph{adaptive} selection of the pertinent points of that grid to retain \cite{SSL}. A reduction of complexity can also be found by leveraging the structure of the equations. A common approach is to formulate the problem with separated variables. A prime example of this approach can be found in the works of Beylkin and Mohlenkamp \cite{beylkin}. The separated variables representation inherently relies on a tensorial factorization of the operators, often based on the use of Kronecker product-based linear systems \cite{greedy-rank-one,LOAN200085}. This reliance on tensor representation can explain the recent development of plenty of methods based on tensor network formalisms, more specifically the Tensor-Train (TT-) format \cite{TTformat} and the overall gain of interest toward tensor methods \cite{kolda,koldabook,hackbush}.

    A popular variable separation technique stemming from computational mechanics is the Proper Generalized Decomposition (PGD) \cite{PGD} and its sparse data-based version, the s-PGD \cite{ibanez}. Those approaches rely on a Canonical Polyadic Decomposition (CPD) of the field that is learnt from one-dimensional modes, \emph{i.e.} a rank-1 tensor is added at each loop on all axes. It resembles the CP-ALS method \cite{kolda} and solves for each mode employing a fixed-point strategy. This fixed-point approach implies a sequential construction and may append modes late in the algorithm that serve only as correctives for early ones that were not optimal.

    Many recent approaches are based on the TT-format and its good scaling with high-dimensional problems. A framework that inspired this paper is that of MANDy \cite{MANDy}. In this work, Gelß et al. adapt the popular SINDy framework \cite{databook} to high-dimensional dynamical system identification, which remains a limitation for SINDy. To achieve good scaling, the authors built sparse TT cores from a choice of one-dimensional basis functions. The same authors later expanded the same methodology to image classification \cite{gelss-klus-image-kernel} where they started mixing their tensorial framework with kernel interpretations to tackle efficiently higher dimensional problems without relying on the TT-SVD sequential algorithm which they demonstrated became a computational bottleneck. At the same time as Baddoo et al. with their fully kernel-based approach to dynamic model learning \cite{LANDO}, Gelß et al. published the kernel Feature Space Approximation (kFSA) method that once again mixes TT-formulated feature maps to kernel methods to provide a high-dimensional regularization (and model reduction) technique. We further investigated the work of Engel et al. in signal processing \cite{KRLS} that develop a kernel recursive least squares algorithm for online learning of signals. In their work, they provide several useful results, noticeably on the limitation of the number of retained signals (finite dictionary) and a bound with respect to $\ell_2$ optimality provided by the kernel Principal Component Analysis (kPCA). The evaluation of distance in the feature space that defines their Almost Linearly Dependent (ALD) criterion is the same later used by Gelß and Baddoo in their respective works. We use the same approach in our framework.
    
    In this work, we keep the philosophy of structuring the equations that build the regression operator from a separated variables perspective. Considering the lack of data as a structural constraint, we isolate it in the structure of the proposed operator and define a \emph{partial} tensor product. We consider it an algebraically minimalistic approach to generalize matrix-based regression. It allows us to preserve properties of usual matrix-based linear systems while avoiding storage of- and computations with the exponentially large operator. By circumventing the fixed point strategy usually encountered in high dimensional contexts, we can produce the minimum-norm least squares solution from available data at a very low cost. We present the approach as \emph{algebraically minimalistic} for its conceptual simplicity compared to tensor networks or abstract kernel feature spaces for instance. Working with tensor algebra and \emph{then} identifying the kernel from the dot products also allows us to avoid the usual limitation of kernel methods to a few well-known kernels while leveraging their properties. For instance we are able to use the properties of kPCA. Furthermore, the sought minimalism allows us to avoid the necessity for sparse operatorsn, required in practice for computational efficiency in the case of TT-formulation. The presented regularizations allow to mitigate the growth of condition number and prove more efficient in high-dimensional contexts (more than 100 parameters) than TT-based techniques.
    
    The proposed approach is structured around a tensor formalism from which we derive nice algebraic properties and computational tricks. We call it the \emph{m-tensor} format as it directly relates to the unmanageable \emph{m}atrix that would be built naively to solve the problem. This is how we retain many of the desirable properties of the matrix-based approach in a fully (tensor-)factorized manner. By preserving a structure of higher order, we enjoy the great efficiency of tensor-based approaches, which would otherwise be lost by the matrix-based structure. As it is tailored for regression operators, this formalism is not investigated as a potential data structure for high order problems. The tools used in this paper are available on GitHub\footnotemark{} along with numerical examples.

    \footnotetext{GitHub repository for the toolbox and examples: \href{this repository}{https://github.com/RemiC-OV/MTensor}.}

    The paper is organized as follows. In \cref{sec:approx_scalar}, we set the problem of the least squares approximation of a scalar function in a general form. \Cref{sec:sep_variables} builds the link between separable variables and tensor structure, this leads us to present the tensor operations we will use within our framework. The m-tensor format is presented in the case of least squares regression in \cref{sec:MT_lstsq}. We provide a general numerical analysis of the method in \cref{sec:numerical_stability}. \Cref{sec:MT_regul} provides different regularization techniques for m-tensor regression. We apply our method to demonstrate scalability and robustness in \cref{sec:numex} on a polynomial benchmark with growing dimension, comparing with existing approaches. \Cref{sec:conclu} draws conclusions and presents perspectives for future works.


    \subsection*{Notations}
    We present the notations adopted in the remainder of the paper. Operation symbols are given with their definition. Working with high dimensional objects relies on burdensome notations in order to disambiguate definitions, therefore we try to strictly stick to conventions given here. Uppercase bold symbols (\emph{e.g.} $\mathbf{T}$) are used for tensors. Uppercase symbols (\emph{e.g.} $M$) are used for matrices. Lowercase bold symbols (\emph{e.g.} $\mathbf{v}$) are used for vectors. The notation $\left[\ \cdot\ \right]_{ijk\dots}$ denotes the $ijk\dots$-th element of the tensor in brackets, each index evolving along an axis of the tensor. Due to the high order nature of the object of the paper, indices are more frequently subscripted with axis number, \emph{i.e.} elements of a tensor $\mathbf{T}$ of order $n$ will be denoted as $[\mathbf{T}]_{i_1i_2\dots i_n}$. We denote the identity of size $m$ as $\Id_m$ and the vector of ones of size $m$ as $\mathds{1}_m$.

\section{Least squares approximation of a scalar function}
\label{sec:approx_scalar}

    In order to provide a general approach, let $\mathcal{D}\subset\mathbb{R}^{d}$ denote a $d$-dimensional domain, defined as the cartesian product of one-dimensional intervals in $\mathbb{R}$, \emph{i.e.} $\mathcal{D}=\times_{i=1}^d[a_i, b_i], \quad [a_i, b_i]\subset\mathbb{R}.$ The aim is to approximate a scalar function $f$ defined over that domain. Let $\varphi_{1\leq i\leq\mathcal{N}}$ denote basis functions to be chosen, we give an approximation of $f$ with the form
    \begin{equation}
        \label{eq:ftilde}
            \hat{f}(\mathbf{x})=\sum_{i=1}^\mathcal{N} c_i\varphi_i(\mathbf{x})=\pmb{\varphi}^\intercal(\mathbf{x})\mathbf{c},
    \end{equation}
    where coefficients $\mathbf{c}$ are the \emph{best fit} (in the least squares sense) for the chosen basis functions. The number of terms $\mathcal{N}$ depends on the choice of bases as well. Let $\mathcal{P}$ denote a set of $m$ distinct samples in $\mathcal{D}$ such that 
    $$\mathcal{P}=\left\{\mathbf{x}_k\in\mathcal{D},\ \forall\ k,\ 1\leq k\leq m\right\}.$$
    Consider we are given the evaluation of $f$ at each sample of $\mathcal{P}$. We denote $\mathbf{y}\in\mathbb{R}^m$ the vector defined as
    $$[\mathbf{y}]_k=f(\mathbf{x}_k)\quad \forall\ k,\ 1\leq k\leq m,\ \mathbf{x}_k\in\mathcal{P}.$$
    Let $\Phi\in\mathbb{R}^{m\times \mathcal{N}}$ denote the evaluation of bases $\pmb{\varphi}$ used in \cref{eq:ftilde} over the set $\mathcal{P}$. With $[\Phi]_k$ denoting the $k$-th row of $\Phi$, we have
    $$[\Phi]_k=\pmb{\varphi}^\intercal(\mathbf{x}_k)\quad \forall\ k,\ 1\leq k\leq m,\ \mathbf{x}_k\in\mathcal{P}.$$
    Then, the least squares solution $\hat{\mathbf{c}}$ is the one that, for any given bases $\pmb{\varphi}$, minimizes the squared $\ell_2$ error with respect to $\mathbf{y}$:
    \begin{equation}
        \hat{\mathbf{c}}=\underset{\mathbf{c}}{\operatorname{argmin}}\left\lVert \sum_{k=1}^m\pmb{\varphi}^\intercal(\mathbf{x}_k)\mathbf{c}-[\mathbf{y}]_k\right\rVert_2^2.
        \label{eq:lstsq}
    \end{equation}
    This minimization problem admits an explicit solution based on whether $\Phi$ has linearly independent columns \cref{eq:lstsqindepcol} or rows \cref{eq:lstsqindeprow} that we build using the Moore-Penrose inverse formulation \cite{penrose}:
    \begin{align}
        \label{eq:lstsqindepcol}
        \hat{\mathbf{c}}&=\left(\Phi^\intercal\Phi\right)^{-1}\Phi^\intercal\mathbf{y},\\
        \label{eq:lstsqindeprow}
        \hat{\mathbf{c}}&=\Phi^\intercal\left(\Phi\Phi^\intercal\right)^{-1}\mathbf{y}.
    \end{align}
    This explicit formulation of least squares solution is used to determine the coefficients $\mathbf{c}$ of the approximation. Without any prior knowledge of $f$, the aim is to define basis functions $\varphi$ such that $f$ is well approximated. To that end, one can choose to tailor the functions in order to minimize $\mathcal{N}$ or provide a rich enough set of bases so that all features of $f$ can be captured. In what follows we choose the latter and show how our tensorial framework alleviates the curse of dimensionality within the regression approach.

\section{Separable variables and tensor structure}
\label{sec:sep_variables}
    
    A popular approach in high dimension is the separation of variables. In this setting, the basis functions $\varphi$ are sought as products of one-dimensional functions. We denote those one-dimensional functions $\psi$ and define multidimensional bases $\varphi$ as
    $$\forall\ i,\ 1\leq i\leq \mathcal{N},\ \varphi_i(\mathbf{x})=\prod_{j=1}^d\psi_{ij}([\mathbf{x}]_j).$$
    Like all combinatorial problems, no standard approach would scale well in order to compute all combinations of one-dimensional basis functions. Therefore we demonstrate the use of a convenient tensor framework. Let $\pmb{\psi}\in\mathbb{R}^{p}$ denote the vector of $p$ one-dimensional basis functions deemed rich enough to represent the problem in one dimension given as
    $$\pmb{\psi}(x)=\begin{bmatrix}
        \psi_1(x) & \psi_2(x) & \cdots &\psi_p(x)
    \end{bmatrix}.$$
    All combinations of one-dimensional bases would be found without redundancy in the elements of the tensor $\pmb{\varphi}$ defined as 
    $$\pmb{\varphi}(\mathbf{x})=\bigotimes_{j=1}^d\pmb{\psi}([\mathbf{x}]_j).$$
    We note that $\pmb{\varphi}\in\mathbb{R}^{p\times p\times\dots\times p}$ is a rank-1 tensor of order $d$. Having the same set of basis functions along each axis is an arbitrary choice made to alleviate notations. Based on this representation, the formulation of the approximation of $f$ in \eqref{eq:ftilde} can be written as the inner product
    \begin{equation}
        \label{eq:tensor_approx}
        \hat{f}(\mathbf{x})=\langle\pmb{\varphi}(\mathbf{x}),\ \mathbf{C}\rangle.
    \end{equation}
    Here $\mathbf{C}\in\mathbb{R}^{p\times p\times \dots\times p}$ is an order $d$ tensor of coefficients. The least squares coefficient tensor $\hat{\mathbf{C}}$ is the one minimizing the $\ell_2$ error with respect to $\mathbf{y}$:
    \begin{equation}
        \label{eq:tensor_minimization}
        \hat{\mathbf{C}}=\underset{\mathbf{C}}{\operatorname{argmin}}\left\lVert \sum_{k=1}^m\langle\pmb{\varphi}(\mathbf{x}_k),\ \mathbf{C}\rangle-[\mathbf{y}]_k\right\rVert_2^2.
    \end{equation}
    If we were to vectorize this tensor formulation, we would find the expression given for \eqref{eq:ftilde} with $\mathcal{N}=p^d$. The most important aspect is therefore the ability to only use factorized formulations for problems that would otherwise be inaccessible computationally and memory-wise. The unmanageable nature of the least squares solution $\hat{\mathbf{C}}$ has motivated research on various tensor formats \cite{koldabook,hackbush,kolda}. The general strategy being to truncate the rank of the solution either in the \emph{canonical} sense of tensor rank, or in the \emph{multilinear} one. The former view of rank is the one used in PGD or CP-ALS while the latter is the one found in tensor trains or Tucker representation. In what follows, our aim is to retain the full least squares approximation in an efficient manner, hence our \emph{optimality} claim.

\section{M-Tensor least squares regression}
\label{sec:MT_lstsq}

    In this section, we define elementary operations to be used in what we refer to as \emph{the m-tensor formalism} to efficiently compute a multidimensional least squares regression. The aim is to make the tensor least squares \eqref{eq:tensor_minimization} manageable in arbitrarily high dimensions in the scarce data context, \emph{i.e.} for underdetermined problems.

    \subsection{Definitions}
        We first define a row-wise tensor product. Further in this paper we refer to it as the \emph{m-tensor product}. It generalizes the face-splitting product \cite{face-splitting} (transposed Khatri-Rao) the same way the tensor product generalizes the matrix Kronecker product. Our approach uses a similar structure as the face-splitting product but preserves its multidimensionality, therefore avoiding the index mapping burden.

        \begin{definition}\emph{(Row-wise tensor product)}.
            For a given set of matrices $\{\Psi_i\in\mathbb{R}^{m\times p_i}\}_{i=1}^d$, we define their row-wise tensor product, building an order $d+1$ tensor $\mathbf{\Phi}\in\mathbb{R}^{m\times p_1\times p_2\times\dots\times p_d}$ as\emph{
            \begin{equation*}
                [\mathbf{\Phi}]_{ki_1 i_2 \dots i_d}=\left[\Lotimes_{j=1}^d\Psi_j\right]_{ki_1 i_2 \dots i_d}=\prod_{j=1}^d \left[\Psi_j\right]_{ki_j}.
            \end{equation*}}
        \end{definition}
        We show in this paper how well-suited tensors built that way are for multilinear regression. We refer to the matrices $\{\Psi_i\}_{i=1}^d$ as \emph{cores} of the \emph{m-tensor} $\mathbf{\Phi}$. We call \emph{rows} the rank-1 tensors resulting from the row-wise tensor product and \emph{columns} the mode-1 fibers of the tensor.
        
        \begin{proposition}
            \label{prop:equiv_matrix}
            The mode-1 unfolding of this tensor, denoted $\mathbf{\Phi}_{(1)}$, results in a dense matrix in $\mathbb{R}^{m\times p_1p_2\dots p_d}$. That unfolding is \emph{exactly} equal to the face-splitting product of its cores defined using notations introduced by Slyusar \cite{face-splitting} as
            $$\mathbf{\Phi}_{(1)}=\Psi_1\Box \Psi_2\Box\dots\Box \Psi_d.$$
        \end{proposition}
        
        \begin{proof}
            This is easily demonstrated element-wise from both definitions. Assuming only $\Psi_1\in\mathbb{R}^{m\times p_1}$ and $\Psi_2\in\mathbb{R}^{m\times p_2}$ without loss of generality, the $(ij)$-th element of the face-splitting product, assuming the index map $j=k_1p_2+k_2$, is given as 
            \begin{equation*}
                [\Psi_1\Box \Psi_2]_{ij}=[\Psi_1]_{ik_1}[\Psi_2]_{ik_2}.
            \end{equation*}
            If the index map's convention is preserved, this corresponds to the m-tensor mode-1 unfolding defined in a standard manner. Generalization to more input matrices builds the index map recursively with the same approach.
        \end{proof}
        We introduce the transposition definition that will later allow us to copy matrix operations formalism with tensors.

        \begin{definition}\emph{(Transposition).}
            \label{def:transposition}
            Let $\mathbf{\Phi}\in\mathbb{R}^{m\times p_1\times p_2\times\dots\times p_d}$ built as the m-tensor product of cores $\{\Psi_j\}_{j=1}^d$, its transpose is given as
            $$\mathbf{\Phi}^\intercal\in\mathbb{R}^{p_d\times p_{d-1}\times\dots\times p_1\times m}, \text{ with }\left[\mathbf{\Phi}^\intercal\right]_{i_di_{d-1}\dots i_1k}=\prod_{j=1}^d\left[ \Psi_j\right]_{ki_j}.$$
        \end{definition}
        A key element of our framework is to be able to efficiently act on a massive amount of combinations of basis functions. To do so we seek to obtain operations equivalent to the matrix ones. We first settle elements towards a m-tensor equivalent to matrix multiplication by adapting the general definition of tensor inner product to our case, \emph{i.e.} by leveraging the efficiency of rank-1 tensor operations.
        \begin{definition}\emph{(Inner product).}
            \label{def:inner_product}
            Let $\pmb{\varphi}_1, \pmb{\varphi}_2\in\mathbb{R}^{p_1\times p_2\times\dots\times p_{d}}$ be two arbitrary tensors. Their inner product is defined as\emph{
            \begin{equation*}
                \langle\pmb{\varphi}_1,\ \pmb{\varphi}_2\rangle=\sum_{i_1=1}^{p_1}\sum_{i_2=1}^{p_2}\dots\sum_{i_d=1}^{p_d}[\pmb{\varphi}_1]_{i_1i_2\dots i_d}[\pmb{\varphi}_2]_{i_1i_2\dots i_d}.
            \end{equation*}}
        \end{definition}
        That computation scales as $\mathcal{O}(p^d)$ considering $\forall i,\ p_i=p$. This definition naturally leads to the norm $\lVert\bphi\rVert=\sqrt{\langle\bphi,\bphi\rangle}$. We identify the specific case of the inner product of rank-1 tensors that we further leverage in the case of m-tensors.
        
        \begin{proposition}
            \label{prop:rank1_inner_product}
            Let $\pmb{\varphi}_1, \pmb{\varphi}_2\in\mathbb{R}^{p_1\times p_2\times\dots\times p_{d}}$ defined as rank-1 tensors, respectively
            $$\pmb{\varphi}_1=\bigotimes_{i=1}^d\mathbf{u}_i^1,\quad\text{and}\quad\pmb{\varphi}_2=\bigotimes_{i=1}^d\mathbf{u}_i^2,$$
            where $\forall\ i,\ 1\leq i\leq d,\ \mathbf{u}^{1,2}_{i}\in\mathbb{R}^{p_i}$. Then the complexity of their inner product falls to $\mathcal{O}(dp)$ considering $\forall i,\ p_i=p$ for ease of notation.
        \end{proposition}

        \begin{proof}
            From the rank-1 structure of $\pmb{\varphi}_1$ and $\pmb{\varphi}_2$ and \cref{def:inner_product}, we have
            $$\langle\pmb{\varphi}_1,\ \pmb{\varphi}_2\rangle=\prod_{i=1}^d(\mathbf{u}_i^1)^\intercal\mathbf{u}_i^2.$$
            This simplifies the $d$-dimensional problem and allows to separately evaluate $d$ vector inner products, hence the $\mathcal{O}(dp)$ complexity scaling.
        \end{proof}
        From that property, we can demonstrate the efficiency of m-tensor operations required for our regression framework. We then develop the generalization of matrix product for tensors through the \emph{d-way contracted product}.

        \begin{definition}\emph{(D-way contracted product).}
            \label{def:dway_product}
            Let $\mathbf{\Phi}_1\in\mathbb{R}^{m_1\times p_1\times p_2\times\dots\times p_{d}}$ and $\mathbf{\Phi}_2\in\mathbb{R}^{m_2\times p_1\times p_2\times\dots\times p_{d}}$, their d-way contracted product, denoted $\ltimes$, is given as
            $$\mathbf{\Phi}_1\ltimes\mathbf{\Phi}_2^\intercal=(\mathbf{\Phi}_1)_{(1)}(\mathbf{\Phi}_2)_{(1)}^{\intercal}$$
        \end{definition}
        In a general setting, that product is computed with complexity $\mathcal{O}(m^2p^d)$, by considering for ease of notation that $m_1=m_2=m$ and $\forall i,\ p_i=p$. We can however give a special case when considering m-tensors.

        \begin{proposition}
            \label{prop:dway_product_mtensors}
            Let $\mathbf{\Phi}_1\in\mathbb{R}^{m_1\times p_1\times p_2\times\dots\times p_{d}}$ and $\mathbf{\Phi}_2\in\mathbb{R}^{m_2\times p_1\times p_2\times\dots\times p_{d}}$ be m-tensors respectively based on cores $\{\Psi_i^1\}_{1\leq i\leq d}$ and $\{\Psi_i^2\}_{1\leq i\leq d}$. Their d-way contracted product is computed as 
            \begin{equation}
                \label{eq:m-tensor-product}
                \mathbf{\Phi}_1\ltimes\mathbf{\Phi}_2^\intercal=\bigodot_{i=1}^d\Psi_i^1(\Psi_i^2)^\intercal.
            \end{equation}
            Considering, for ease of notation that $m_1=m_2=m$ and $\forall i,\ p_i=p$, the complexity of this product is $\mathcal{O}(m^2dp)$.
        \end{proposition}

        \begin{proof}
            Starting from the general case of d-way contracted product with a tensor inner product oriented view, we show how we can factorize each dimension and make the element-wise product appear. An element of $\mathbf{\Phi}_1\ltimes\mathbf{\Phi}_2^\intercal\in\mathbb{R}^{m_1\times m_2}$ is given as
            \begin{align*}
                \left[\mathbf{\Phi}_1\ltimes\mathbf{\Phi}_2^\intercal\right]_{k_1k_2}&=\langle[\mathbf{\Phi}_1]_{k_1},\ [\mathbf{\Phi}_2]_{k_2}\rangle,\\
                &=\sum_{j_1=1}^{p_1}\sum_{j_2=1}^{p_2}\dots\sum_{j_d=1}^{p_d}\left[\mathbf{\Phi}_1\right]_{k_1j_1j_2\dots j_d}\left[\mathbf{\Phi}_2\right]_{k_2j_1j_2\dots j_d},\\
                &=\sum_{j_1=1}^{p_1}\sum_{j_2=1}^{p_2}\dots\sum_{j_d=1}^{p_d}\prod_{i=1}^d\left[\Psi^1_i\right]_{k_1j_i}\left[\Psi^2_i\right]_{k_2j_i},\\
                &=\sum_{j_1=1}^{p_1}\left[\Psi^1_1\right]_{k_1j_1}\left[\Psi^2_1\right]_{k_2j_1}\dots\sum_{j_d=1}^{p_d}\left[\Psi^1_d\right]_{k_1j_d}\left[\Psi^2_d\right]_{k_2j_d},\\
                &=\prod_{i=1}^d\sum_{j_i=1}^{p_i}\left[\Psi^1_i\right]_{k_1j_i}\left[\Psi^2_i\right]_{k_2j_i},
            \end{align*}
            which is the element-wise description of \cref{eq:m-tensor-product}. The inner products of rank-1 tensors scales as $\mathcal{O}(dp)$ as proven for \cref{prop:rank1_inner_product}, considering $m_1=m_2=m$ we have $m^2$ of those products to compute hence the complexity $\mathcal{O}(m^2dp)$ considering $\forall\ i,\  1\leq i\leq d,\ p_i=p$.
        \end{proof}

    \subsection{Least squares approximation}
        From those definitions we can build our framework. Provided the set of sampled parameters $\mathcal{P}=\{\mathbf{x}_k\}_{k=1}^m$ and corresponding evaluations of $f$ in a vector $\mathbf{y}$, we define the basis functions tensor as
        $$\pmb{\varphi}(\mathbf{x})=\bigotimes_{i=1}^d\pmb{\psi}([\mathbf{x}]_i),\quad\text{with}\quad\pmb{\psi}:x\mapsto\begin{bmatrix}
            \psi_1(x)&\psi_2(x)&\dots&\psi_p(x)
        \end{bmatrix}.$$
        This structure amounts to combining one-dimensional bases $\psi_{1\leq i\leq p}$ for all input parameters. This yields a large approximation space, unreachable with a matrix-based regression. We define the cores
        \begin{equation*}
            \forall i,\ 1\leq i\leq d,\quad\Psi_i=
            \begin{bmatrix}
                \psi_1([\mathbf{x}_1]_i) & \psi_2([\mathbf{x}_1]_i) & \dots & \psi_p([\mathbf{x}_1]_i)\\
                \psi_1([\mathbf{x}_2]_i) & \psi_2([\mathbf{x}_2]_i) & \dots & \psi_p([\mathbf{x}_2]_i)\\
                \vdots & \vdots & \ddots & \vdots \\
                \psi_1([\mathbf{x}_m]_i) & \psi_2([\mathbf{x}_m]_i) & \dots & \psi_p([\mathbf{x}_m]_i)
            \end{bmatrix}.
        \end{equation*}
        Mappings can be defined differently for each variable. Note that by providing nonlinear mappings, the (multi-)linear system we build is able to approximate nonlinear models in a non-iterative manner. That definition of the cores allows us to construct the m-tensor regression operator $\mathbf{\Phi}$ as
        \begin{equation}
            \label{eq:tensorphi}
            \mathbf{\Phi}=\Lotimes_{i=1}^d \Psi_i.
        \end{equation}
        We can now reformulate the regression problem in m-tensorial setting. We seek the regression coefficients $\mathbf{C}\in\mathbb{R}^{p\times p\times\dots\times p}$ such that
        \begin{equation}
            \label{eq:tensorreg}
            \mathbf{\Phi}\ltimes\mathbf{C}=\mathbf{y}.
        \end{equation}
        A least squares solution $\mathbf{\hat{C}}$ to this problem can be built explicitly based on the general pseudoinverse formulation \cref{eq:lstsqindeprow}. Provided that $\mathbf{\Phi}$ has linearly independent rows, we have the solution
        \begin{equation}
            \label{eq:tensorsolution}
            \hat{\mathbf{C}}=\mathbf{\Phi}^\intercal\left(\mathbf{\Phi}\ltimes\mathbf{\Phi}^\intercal\right)^{-1}\mathbf{y}.
        \end{equation}
        We made our approach fully equivalent to the matrix framework while taking advantage of the tensor factorization, therefore that definition holds in any dimension. We consider the linearly independent rows case because, in general, the exponential growth of the system with $d$ will prevent it to have linearly independent columns as it would require $m$ to grow exponentially with $d$. This makes our approach well defined in the under-determined case. 
        
        We generally cannot use coefficients formulated as \cref{eq:tensorsolution}. For that reason, we readily build the approximation of $f$ as 
        \begin{equation}
            \hat{f}(\mathbf{x})=\pmb{\varphi}(\mathbf{x})\ltimes\mathbf{\Phi}^\intercal\left(\mathbf{\Phi}\ltimes\mathbf{\Phi}^\intercal\right)^{-1}\mathbf{y}.
        \end{equation}
        This expression shows explicitly how it results only in efficient operations. In fact, by precomputing 
        \begin{equation}
            \mathbf{z}=(\mathbf{\Phi}\ltimes\mathbf{\Phi}^\intercal)^{-1}\mathbf{y}\in \mathbb{R}^{m},
        \end{equation}
        we can evaluate $\hat{f}$ by computing $\pmb{\varphi}(\mathbf{x})\ltimes\mathbf{\Phi}^\intercal\in\mathbb{R}^{m}$ and its inner product with $\mathbf{z}$:
        \begin{equation}
            \hat{f}(\mathbf{x})=\langle\pmb{\varphi}(\mathbf{x})\ltimes\mathbf{\Phi}^\intercal,\ \mathbf{z}\rangle.
        \end{equation}
        This evaluation's complexity scales as $\mathcal{O}(mdp)$. The result provided is optimal in the least squares sense for the selected mappings and given sampled points.

    \subsection{Numerical analysis}
        \label{sec:numerical_stability}
        In this section we analyze the influence of the structure we propose on roundoff error, noise sensitivity and conditioning. For the setting, we consider that our regression operator $\mathbf{\Phi}$ is built upon a sample $\mathcal{P}$ of $m$ elements and a tensorized mapping $\pmb{\varphi}=\bigotimes_{i=1}^d\pmb{\psi}_i$. We assume the mappings $\{\pmb{\psi}_i\in\mathbb{R}^{p}\}_{i=1}^d$ to be able to capture $f$, \emph{i.e.}
        $$\exists\ \mathbf{C}\in\mathbb{R}^{p\times p\times \dots\times p},\ \forall\ \mathbf{x}\in\mathbb{R}^d,\quad f(\bx)=\pmb{\varphi}(\mathbf{x})\ltimes\mathbf{C}.$$
        All proofs are developed in \cref{sec:appendix:stability}. We consider a vector $\mathbf{y}$ storing evaluations of $f$ over $\mathcal{P}$. We denote the approximation we build $\tilde{f}(\bx)=\bphi(\bx)\ltimes(\bC+\dbC)$, \emph{i.e.} the approximation stems from a perturbation of the coefficient tensor. We will consider input perturbation $[\dbx]_{1\leq i\leq d}=\mathcal{O}(\epsilon_\text{machine})$ to evaluate its effect of the stability of our approach. In what follows, we consider approximations by neglecting at least doubly infinitesimal quantities, \emph{i.e.} we consider first order approximations.
        \begin{theorem}\emph{(Bound on $\bphi(\bx+\dbx)$).}
            \label{th:bound_phixdx}
            Let $\pmb{\xi}_{i=1}^d:\mathbb{R}^d\times\mathbb{R}^d\rightarrow\mathbb{R}^{p}$ denote a 1D mapping error due to input perturbation such that the separated variables representation yields
            \begin{equation*}
                \bphi(\bx+\dbx)=\bigotimes_{i=1}^d\bpsi_i([\bx]_i+[\dbx]_i)=\bigotimes_{i=1}^d\biggl(\bpsi_i([\bx]_i)+\pmb{\xi}_i([\bx]_i, [\dbx]_i)\biggr).
            \end{equation*}
            Let the 1D mappings $\bpsi_{i=1}^d$ be normalized element-wise, then the evaluation error of $\bphi$ is bounded as
            \begin{equation*}
                \lVert\bphi(\bx+\dbx)-\bphi(\bx)\rVert_F\lesssim\sqrt{p^{d-1}}\sum_{i=1}^d\lVert\pmb{\xi}_i([\bx]_i,[\dbx]_i)\rVert_F.
            \end{equation*}
        \end{theorem}
        \Cref{th:bound_phixdx} demonstrates the necessity to analyze the behaviour of the terms $\pmb{\xi}$ when choosing a mapping so as to minimize input perturbation error. The specific case of \emph{input roundoff error} gives \cref{lem:inputroundoff}.
        \begin{lemma}\emph{(Input roundoff error).}
            \label{lem:inputroundoff}
            Let each input $\bx$ of $\bphi$ be perturbed with $[\dbx]_{1\leq i\leq d}=\mO(\epsilon_\text{machine})$, then there exists a mapping-dependent constant $c$ such that
            $$\lVert\bphi(\bx+\dbx)-\bphi(\bx)\rVert_F\leq dc\epsilon_\text{machine}\sqrt{p^{d}}.$$
        \end{lemma}
        We draw conclusions on the backward stability of the approach under those assumptions and derive strategies to guarantee stability.
        \begin{theorem}\emph{(Backward stability).}
            \label{theorem:backward_stability}
            Let $\bpsi_{i=1}^d$ denote normalized 1D mappings with separable perturbation-induced errors of the shape $\bpsi_i(x+\delta x)=\bpsi_i(x)+\pmb{\xi}_i(x,\delta x)$. Let $c$ denote a constant such that over its domain of definition 
            $$\lVert\pmb{\xi}_i(x,\delta x)\rVert_F\leq c\sqrt{p}.$$
            Then if $dc\ll 1$, the backward error tends to
            $$\dfrac{\lVert\tilde{f}(\bx)-f(\bx+\dbx)\rVert}{\lVert f(\bx+\dbx)\rVert}\lesssim\dfrac{\lVert\dbC\rVert}{\lVert\bC\rVert}.$$
        \end{theorem}
        \Cref{theorem:backward_stability} tells us that controlling the bounding constant $c$, which is directly derived from the choice of 1D mapping, we can guarantee that the backward error tends purely to the approximation error of the model $\dbC$. The advantage is that a uniform scaling of the bases suffices to lower that constant without impacting the conditioning (for which we design specific regularization strategies in \cref{sec:MT_regul}).
        
        The aim of our method is to preserve the matrix properties of its mode-1 unfolding $\mathbf{\Phi}_{(1)}\in\mathbb{R}^{m\times p^d}$. Thus we may consider the SVD of that unfolding
        \begin{equation}
            \label{eq:Phi1SVD}
            \mathbf{\Phi}_{(1)}=U\Sigma V^\intercal.
        \end{equation}
        We can access $U$ and $\Sigma$ through the eigendecomposition of
        \begin{equation}
            \label{eq:PSVD}
            P=\mathbf{\Phi}\ltimes\mathbf{\Phi}^\intercal=\mathbf{\Phi}_{(1)}\mathbf{\Phi}_{(1)}^\intercal=U\Sigma^2 U^\intercal.
        \end{equation}
        This gives us access to the effective condition number of $\mathbf{\Phi}_{(1)}$ given by
        \begin{equation}
            \label{eq:cond}
            \kappa(\mathbf{\Phi})=\dfrac{\sigma_1}{\sigma_r},
        \end{equation}
        where $\sigma_1$ denotes the largest singular value in $\Sigma$ and $\sigma_r$ its lowest nonzero singular value. Then we can assess noise sensitivity by assuming the perturbations $[\delta\mathbf{y}]_{k=1}^m$.
        \begin{theorem}\emph{(Noise sensitivity).}
            \label{th:bound_noise}
            Assuming the sampled points $\mP$ are sufficient and that the basis $\bphi$ is well suited to approximate $f$, then the error stemming from data noise $\delta\mathbf{y}$ is bounded as 
            \begin{equation*}
                \lVert\dbC\rVert\leq\dfrac{\lVert\delta\mathbf{y}\rVert}{\sigma_r}.
            \end{equation*}
        \end{theorem}
        This bound demonstrates the crucial necessity to regularize the regression.

 
\section{Regularization of m-tensor regression}
\label{sec:MT_regul}
    In this section we adapt regularization techniques for m-tensor regression. In under-determined regression with $m$ samples, the rank is at most $m$ and the condition number is an indicator on \emph{row correlation}. In the case of m-tensor regression, this means measuring \emph{how close to linearly dependent} the mapped samples $\{\bphi(\bx_k)\}_{k=1}^m$ are. In this section, we denote $P=\mathbf{\Phi}\ltimes\mathbf{\Phi}^\intercal$.
    
    \subsection{Tikhonov regularization}
        Analogously to the matrix case, Tikhonov regularization for m-tensors is given as
        $$\hat{\mathbf{C}}_2=\mathbf{\Phi}^\intercal\left(P+\lambda^2\Id\right)^{-1}\mathbf{y}.$$
        This results in a uniform lifting of the eigenvalues of $P$, hence the conditioning improvement. The approximation of $f$ with Tikhonov regularization is then given by evaluating the inner product
        \begin{equation}
            \label{eq:l2appx}
            \hat{f}_2(\mathbf{x})=\langle\pmb{\varphi}(\mathbf{x})\ltimes\mathbf{\Phi}^\intercal,\ \mathbf{z}_2\rangle\ \text{with}\ \mathbf{z}_2=\left(P+\lambda^2\Id\right)^{-1}\mathbf{y}.
        \end{equation}
        The complexity of this approach is dominated by the inversion scaling as $\mathcal{O}(m^3)$ and the construction of $P$ itself, scaling as $\mathcal{O}(m^2pd)$. Inference, given in \eqref{eq:l2appx}, has a complexity $\mathcal{O}(mpd)$.

    \subsection{Spectral truncation}
    
        We manage to always take advantage of the high efficiency of the inner product of rank-1 tensors. In the under-determined case, we make $P$ appear, with $[P]_{ij}=\langle\varphi(\mathbf{x}_i), \varphi(\mathbf{x}_j)\rangle.$ Thus, considering the mapping
        $$\varphi:\mathbf{x}\longmapsto\bigotimes_{i=1}^d \pmb{\psi}([\mathbf{x}]_i),$$
        we can derive a kernel function $k$ defined as the inner product
        $$k:(\mathbf{x}_i,\ \mathbf{x}_j)\longmapsto\langle\pmb{\varphi}(\mathbf{x}_i),\ \pmb{\varphi}(\mathbf{x}_j)\rangle.$$
        Therefore we are guaranteed to build a symmetric positive semidefinite kernel. Assuming the 1D mappings $\bpsi$ to be continuous on compact domains (such as intervals of $\R$), then the kernel we define satisfies the assumptions \cite{mercer} under which Mercer-type spectral results apply. This allows us to leverage properties of kernels \cite{smolakernel}. Moreover it makes spectral truncation within the m-tensor framework a truncated kernel PCA \cite{kPCA}. Keeping only the first $r$ principal components in feature space assures us to preserve the optimal $r$-dimensional representation in feature space for a given mapping.

        Spectral truncation improves the conditioning of the regression operator by setting to zero the singular values below a set threshold $\tau$ or up to a prescribed rank $r$. We apply it to m-tensor regression through the eigenvalues of $P$. Let $U_r$ denote the matrix built from the first $r$ eigenvectors of $P$. The regularized coefficients are then given by
        $$\hat{\mathbf{C}}_*=\mathbf{\Phi}^\intercal U_r\Sigma_r^{-2}U_r^\intercal\mathbf{y}.$$
        The resulting approximation is computed as the inner product
        \begin{equation}
            \label{eq:nuclearappx}
            \hat{f}_*(\mathbf{x})=\langle\pmb{\varphi}(\mathbf{x})\ltimes\mathbf{\Phi}^\intercal,\ \mathbf{z}_*\rangle,\ \text{with}\ \mathbf{z}_*=U_r\Sigma_r^{-2}U_r^\intercal\mathbf{y}.
        \end{equation}
        This approach results in a complexity dominated by the computation of the SVD of $P$, \emph{i.e.} $\mO(m^3pd)$ with the initial cost $\mO(m^2pd)$ to build $P$. Inference cost is the same as for Tikhonov regularization.

    \subsection{Subsampling}

        We consider a regularization technique that aims at both enhancing the conditioning of the system and improving inference's efficiency. We consider a subset of $\tilde{m}<m$ rows of $\mathbf{\Phi}$, corresponding to a subsample $\tilde{\mathcal{P}}\subset\mathcal{P}$ that can approximate all rows of $\mathbf{\Phi}$ up to a given tolerance $\varepsilon$, \emph{i.e.}
        \begin{equation}
            \label{eq:subsampling}
            \exists W\in\mathbb{R}^{m\times\tilde{m}},\quad\left\lVert\mathbf{\Phi}-W\tilde{\mathbf{\Phi}}\right\rVert_2^2<\varepsilon.
        \end{equation}
        Rows of $\mathbf{\Phi}$ are considered Almost Linearly Dependent (ALD) \cite{KRLS} on this subset. We denote $\tilde{\mathbf{\Phi}}$ the tensor made only of those $\tilde{m}$ rows of $\mathbf{\Phi}$. Those are selected based on the following squared distance $\delta_k$ evaluation. By denoting $[\mathbf{\Phi}]_k$ the $k$-th row of $\mathbf{\Phi}$,
        $$\forall k,\ 1\leq k\leq m,\quad\delta_k=\underset{\mathbf{w}\in\mathbb{R}^{\tilde{m}}}{\operatorname{min}}\left\lVert[\mathbf{\Phi}]_k-\mathbf{w}^\intercal\tilde{\mathbf{\Phi}}\right\rVert_2^2<\varepsilon.$$
        $\delta_k$ evaluates the distance of each row $[\mathbf{\Phi}]_k$ to the subspace spanned by the rows of $\tilde{\mathbf{\Phi}}$. $W$ are then the coordinates of the approximated rows restricted to that span. Decomposition \cref{eq:subsampling} provides a bound with respect to the $\ell_2$ optimal provided by the previously introduced spectral truncation \cite{KRLS}. Sampling from the rows preserves the rank-1 structure, which is lost with spectral truncation. By separating the subset $\tilde{\mathbf{\Phi}}$ and the remainder that we denote $\bar{\mathbf{\Phi}}$ in \cref{eq:tensorreg} and applying the same separation and reordering in $\mathbf{y}$, we have
        $$\begin{bmatrix}
            \tilde{\mathbf{\Phi}}\\
            \bar{\mathbf{\Phi}}
        \end{bmatrix}\ltimes\mathbf{C}=\begin{bmatrix}
            \tilde{\mathbf{y}}\\
            \bar{\mathbf{y}}
        \end{bmatrix}.$$
        We then only solve for the first subset since the remainder can \emph{almost} be linearly deduced from it with the weights $W$, \emph{i.e.} we seek $\tilde{\mathbf{C}}$ in the minimization problem 
        \begin{equation*}
            \label{eq:ALIminimization}
            \tilde{\mathbf{C}}=\underset{\mathbf{C}}{\operatorname{argmin}}\left(\lVert \tilde{\mathbf{\Phi}}\ltimes\mathbf{C}-\tilde{\mathbf{y}}\rVert_2^2\right).
        \end{equation*}
        The conditioning is enhanced by removing \emph{close-to-correlated} samples. It also truncates the spectrum so $\operatorname{rank}(\tilde{\bPhi})=\tilde{m}$. This subsystem approach is the only regularization that manages to reduce the size of the system to solve and provide a gain with respect to $m$ along with the linearization of inference complexity with $d$.

        The subset selection can be done in one of two ways: a greedy or streamed one. The former is developed analogously to what is done by Gelß et al. \cite{kFSA} while the latter is an adaptation of the kernel recursive least squares developed by Engel et al. \cite{KRLS} for signal processing. The approach proposed by Engel et al. is enhanced by incrementally building the Cholesky factor instead of relying on the kernel matrix, as inspired by Baddoo et al. \cite{LANDO}. The approximation is evaluated as
        $$\tilde{f}(\mathbf{x})=\langle\pmb{\varphi}(\mathbf{x})\ltimes\tilde{\mathbf{\Phi}}^\intercal,\ \tilde{\mathbf{z}}\rangle,\quad\text{with}\quad\tilde{\mathbf{z}}=(\tilde{\mathbf{\Phi}}\ltimes\tilde{\mathbf{\Phi}}^\intercal)^{-1}\tilde{\mathbf{y}}.$$
        The computational gain at the online stage can reveal massive compared to the least squares formulation if $\tilde{m}\ll m$. With the subsampling approach, inference cost scales as $\mathcal{O}(\tilde{m}pd)$. This formulation outperforms all previous ones at a small and controllable cost in terms of error. Algorithms for the greedy and streamed subsampling formulations are given in appendix \cref{sec:appendix:subsample}. The streamed construction results in a complexity bounded by $\mathcal{O}(m^2pd)$ by considering $\forall\ i,\ p_i=p$.

    \subsection{Summary}

        \Cref{tab:complexity_recall} summarizes the theoretical computational complexity of each regression approach at construction and inference stages. We have considered a number of samples scaling linearly with the dimension $d$ with a factor $\alpha$, so $m=\alpha d$.

        \begin{table}[ht!]
            \caption{Theoretical complexity evaluations for construction and inference for different regressions.\label{tab:complexity_recall}}
            \begin{center}
                \begin{tabular}{c|c|c|c}
                        & Least squares & Spect. Trunc. & Subsampling \\ \hline\hline
                    Construction & $\mathcal{O}(\alpha^3d^3p)$ & $\mathcal{O}(\alpha^3d^3p)$ & $\mathcal{O}(\alpha^2d^3p)$ (up. bound) \\ \hline
                    Inference & $\mathcal{O}(\alpha d^2p)$ & $\mathcal{O}(\alpha d^2p)$ & $\mathcal{O}(\tilde{m}dp)$
                \end{tabular}
            \end{center}
        \end{table}

\section{Numerical example}
\label{sec:numex}

    We demonstrate the scaling of our approach on the Rosenbrock function. It is defined for arbitrary dimension $d$ over the domain $\mathcal{D}=[-5, 10]^{\times d}$ as
    \begin{equation*}
        f(\bx)=\sum_{i=1}^{d-1}100\left([\bx]_{i+1}-[\bx]_i^2\right)^2+\left([\bx]_i-1\right)^2.
    \end{equation*}
    \paragraph{Proposed model}We approximate the Rosenbrock function with growing dimension $d\in\{10, 20, 30, 50, 100, 200\}$ with polynomial mappings of degree 4. Then we evaluate conditioning when augmenting the degree of approximation, \emph{i.e.} varying $q$ in the mapping defined as
    \begin{equation*}
        \pmb{\varphi}(\bx)=\bigotimes_{i=1}^d\begin{bmatrix}
            1 & c[\bx]_i & c[\bx]_i^2 & \dots & c[\bx]_i^q
        \end{bmatrix}.
    \end{equation*}
    The scaling factor $c$ is computed following the reasoning of \cref{sec:appendix:poly_stability}, here $c=(q10^{q})^{-1}$. The training set $\mathcal{P}$ of $m=\alpha d$ elements for $\alpha=150$, is drawn randomly from a Latin Hypercube Sample (LHS) of $5m$ elements, \emph{i.e.} we consider linear scaling of the sample with dimension so as to remain in the underdetermined regime. The m-tensor regression is compared to the polynomial kernel
    \begin{equation}
        k(\bx, \mathbf{y})=(1+\bx^\intercal\mathbf{y})^q.
    \end{equation}
    This kernel is sufficient to capture the function. We use it in vanilla kernel regression built as
    \begin{equation}
        \tilde{f}(\bx)=k(x,\mathcal{P})\mathbf{z},\quad\text{with}\quad\mathbf{z}=K(\mathcal{P}, \mathcal{P})^{-1}\mathbf{y}.
    \end{equation}
    The same kernel is used in the kernel Feature Space Approximation (kFSA) method \cite{kFSA} which builds greedily a subsample $\tilde{\mathcal{P}}$ based on a threshold distance in feature space just as our subsampling approach but with an implicit feature map. It is compared with the streamed subsampling approach. We do not include MANDy-like regression as is revealed too heavy on memory. We used the implementations of P. Gelß for MANDy and kFSA as references\footnote{GitHub account for the reference codes: \href{this repository}{https://github.com/PGelss}.}. The chosen kernel captures only the polynomial terms that are present in the Rosenbrock function whereas the mapping $\pmb{\varphi}$ captures more terms, therefore, we should expect slightly worse results for m-tensor regression based on this mapping.

    \paragraph{Error evaluation}The error is evaluated from a test set of $2m$ unseen points from the original LHS. The error is computed as the $\ell_2$ norm of errors over the testing set of points. Let $\mathbf{y}=\begin{bmatrix}
        y_1 & y_2 & \dots & y_{2m}
    \end{bmatrix}$ denote the vector of exact outputs for each input in the testing set. Its approximated counterpart is denoted $\tilde{\mathbf{y}}$, the error we observe is 
    $$E=\dfrac{\lVert \mathbf{y}-\tilde{\mathbf{y}}\rVert_2}{\lVert \mathbf{y}\rVert_2}.$$

    \paragraph{Results}\Cref{fig:ex:rosenbrock} displays the results of the benchmark we conducted in terms of building CPU time, inference time and error of approximation on an unseen test set. In \cref{tab:ex:rosenbrock} we give the thresholds selected for each approach and the corresponding number of retained samples (or rank for spectral truncation). Note in the case of spectral truncation that the truncation rank $r$ is selected as
    $$r=\underset{i}{\operatorname{argmin}}\left(1-\dfrac{\sum_{j=1}^i\sigma_j^2}{\sum_{j=1}^m\sigma_j^2}\right)\quad\text{st.}\quad1-\dfrac{\sum_{j=1}^i\sigma_j^2}{\sum_{j=1}^m\sigma_j^2}\leq\varepsilon.$$
    Thresholds for distances in feature spaces were selected to minimize the building error without keeping the full sample\footnote{We did not automate nor optimized the thresholds definitions, rather we fitted thresholds as powers of 10 and barely adjusted.}, otherwise least squares m-tensor regression and kernel regression would be more pertinent.

    \begin{figure}[ht!]
        \centering
        \begin{subfigure}[t]{.5\linewidth}
            \centerline{\includegraphics[width=\textwidth]{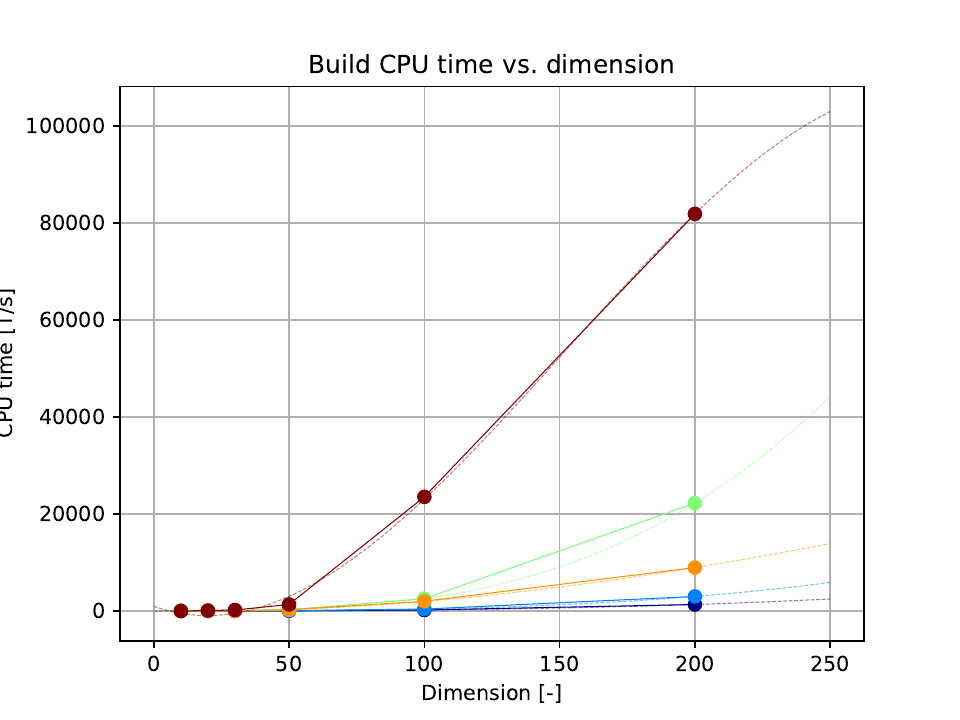}}
            \caption{Average CPU time to build.}
            \label{fig:ex:rosenbrock:build_time}
        \end{subfigure}%
        \begin{subfigure}[t]{.5\linewidth}
            \centerline{\includegraphics[width=\textwidth]{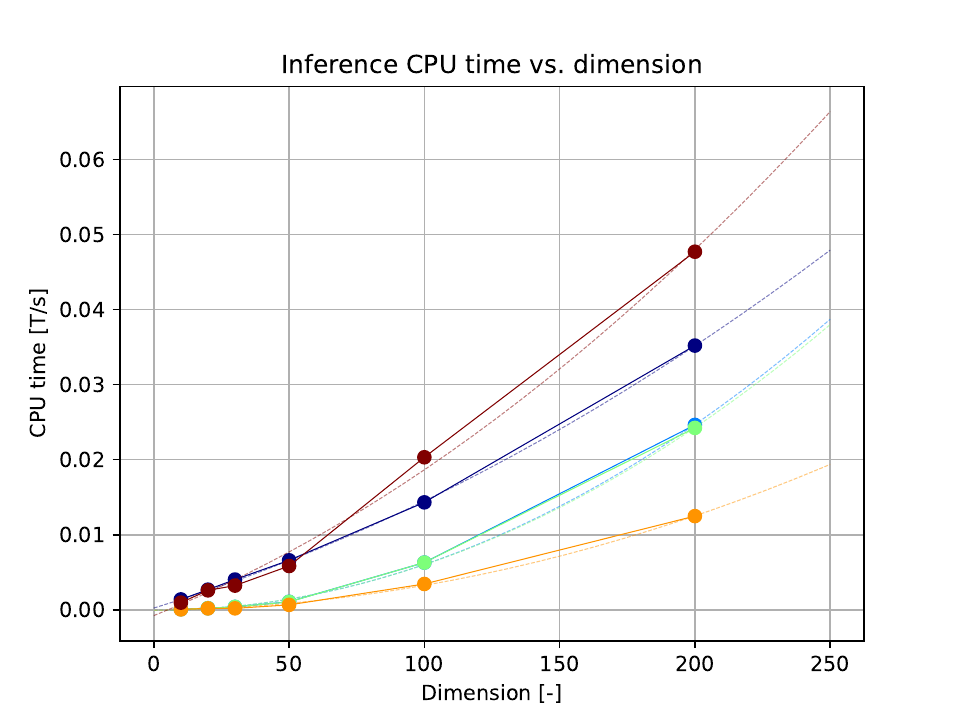}}
            \caption{Average CPU time to evaluate.}
            \label{fig:ex:rosenbrock:inference_time}
        \end{subfigure}
        \begin{subfigure}[t]{.65\linewidth}
            \centerline{\includegraphics[width=\textwidth]{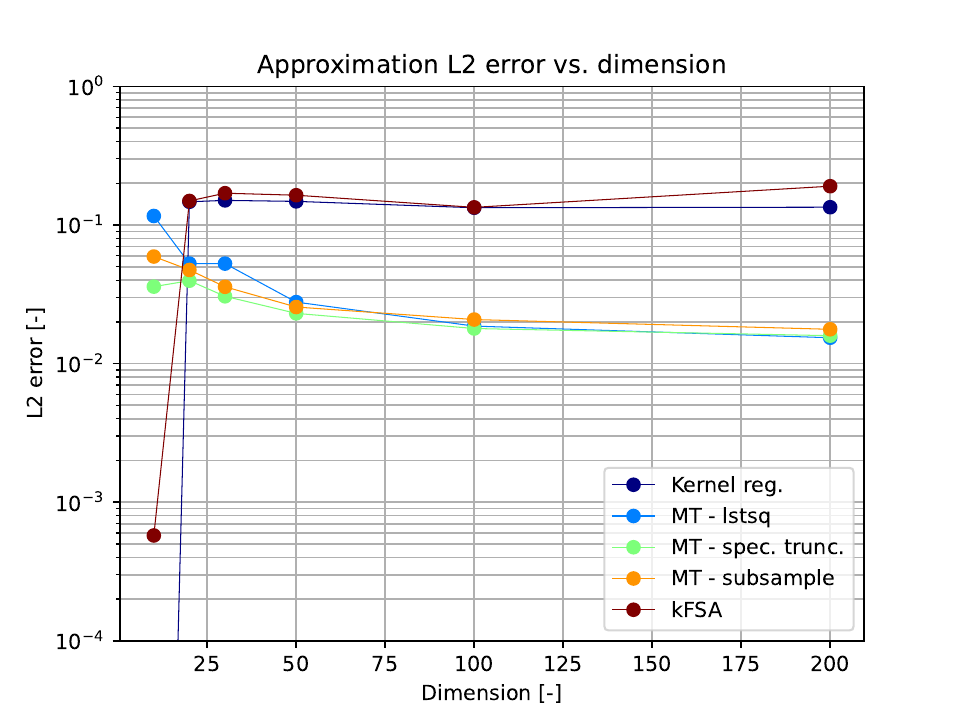}}
            \caption{Average model error on test set.}
            \label{fig:ex:rosenbrock:error}
        \end{subfigure}
        \caption{Benchmark of m-tensor regression with kernel regression and kFSA on the Rosenbrock function with varying dimension (dashed lines are cubic and quadratic fits for the building and inference times resp.).}
        \label{fig:ex:rosenbrock}
    \end{figure}

    \begin{table}[ht!]
        \caption{Hyperparameter values and average number of retained samples (or rank) for kFSA, subsampling and spectral truncation\label{tab:ex:rosenbrock}}
        \begin{center}
        \begin{tabular}{l|c|c|c|c|c|c}
            Dimension & 10 & 20 & 30 & 50 & 100 & 200 \\ \hline\hline
            $\tilde{m}$ kFSA & 985 & 2724 & 3371 & 6092 & 14927 & 19461 \\ \hline
            $\tilde{m}$ subsampl. & 340 & 2305 & 1296 & 4581 & 8896 & 14397 \\ \hline
            $r$ spec. trunc. & 260 & 1731 & 2195 & 4027 & 9359 & 19180 \\ \hline \hline
            Thresh. kFSA & $10^{4}$ & $10^{9}$ & $10^{11}$ & $10^{12}$ & $10^{13}$ & $5\cdot10^{14}$ \\ \hline
            Thresh. subsampl. & $10^{-6}$ & $10^{-6}$ & $10^{-4}$ & $10^{-3}$ & $2\cdot10^{-2}$ & $2.5\cdot10^{-1}$ \\ \hline
            Thresh. spec. trunc. & $10^{-7}$ & $10^{-7}$ & $10^{-5}$ & $10^{-4}$ & $10^{-3}$ & $10^{-2}$
        \end{tabular}
        \end{center}
    \end{table}
    We demonstrate that our approach outperforms other methods when the dimension grows dramatically. The approach of kFSA relies on a threshold that is harder to interpret since it acts on an implicit feature map, it is easy to settle for one that lowers the number of retained samples but error explodes rapidly (even for comparable amounts to our streamed subsampling approach). If many are kept, we end up confronted to the same result as kernel regression in high dimensions, with a higher overhead to compute. We plotted in dashed lines the best fit for each evolution with the degree given in \cref{tab:complexity_recall} to validate our complexity analysis.

    \begin{figure}[ht!]
        \centering
        \begin{subfigure}[t]{.5\linewidth}
            \centerline{\includegraphics[width=\textwidth]{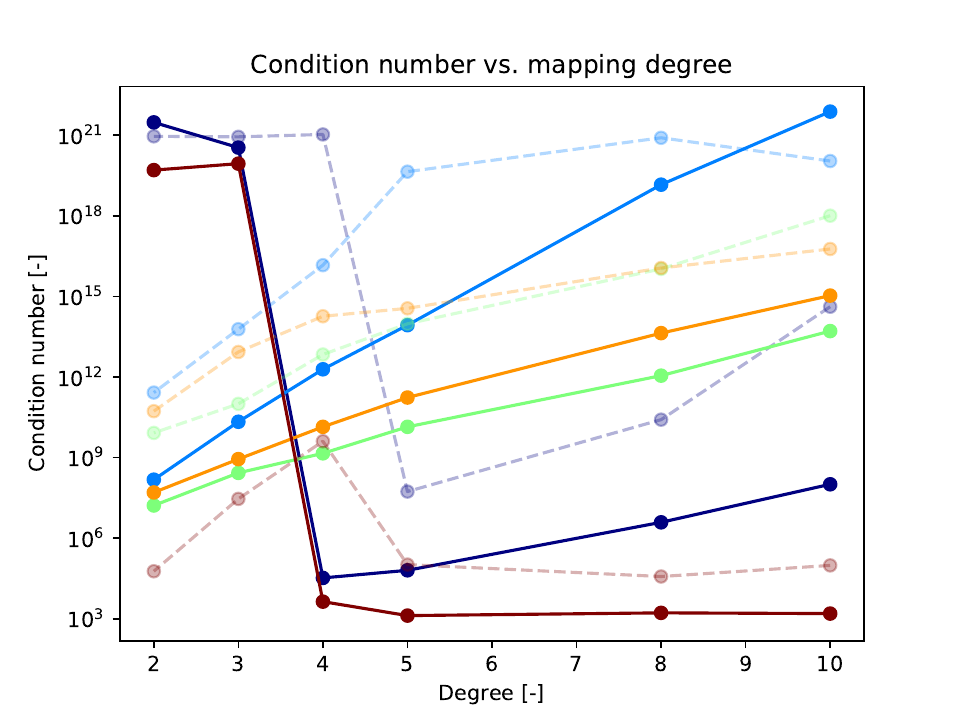}}
            \caption{Average condition number of Gram matrices.}
            \label{fig:ex:cond:kappa}
        \end{subfigure}%
        \begin{subfigure}[t]{.5\linewidth}
            \centerline{\includegraphics[width=\textwidth]{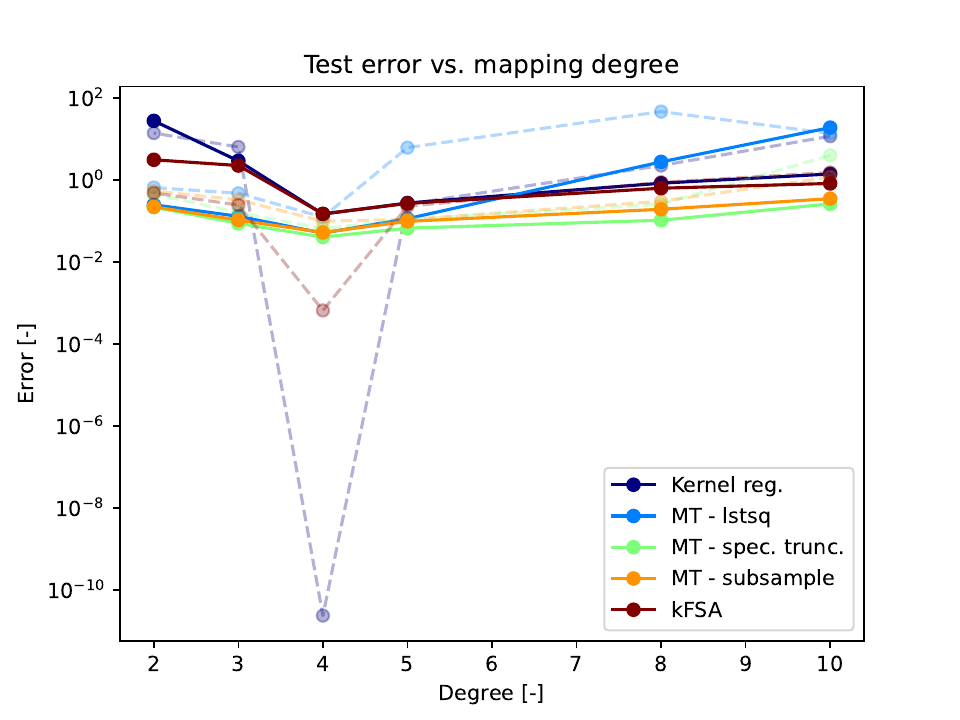}}
            \caption{Average error.}
            \label{fig:ex:cond:error}
        \end{subfigure}
        \caption{Evaluation of the conditioning of the various approaches with respect to the mapping's or kernel's degree (continuous line is $d=20$ and dashed $d=10$).}
        \label{fig:ex:cond}
    \end{figure}
    \Cref{fig:ex:cond:kappa} displays the average condition number of the final Gram matrices to be inverted in each method with varying order. \Cref{fig:ex:cond:error} shows the associated average error. The chosen parameters can be found with the code for the examples. Apart from low orders $q\in\{2,3\}$, m-tensor regression has a much worse conditioning than kernel regression but the proposed regularizations are effective and lower it by orders of magnitude. The jump in conditionning for kernel-based approaches is due to a switch between the oversampled and the undersampled regime based on the actual dimension of the feature space for the polynomial kernel: $\binom{d+q}{q}$.

    \paragraph{Convergence}We now take interest in the convergence of the model with the size of sample. To that end, we consider 5 models built on LHS of $m=\alpha d$ elements with $d=100$ and $\alpha\in\{50, 100, 150, 200, 250\}$. Each of those models is then tested on 2 unseen LHS of $m=\alpha d$ elements. We consider m-tensor regression only as the convergence of kernel regression is an already well-studied problem \cite{kernelconvergence}. It should be noted that in comparison with MANDy (function major), kFSA and kernel regression, m-tensor regression aims at building a very large number of features to regress with, as is done with coordinate major MANDy. The features produced by kernel approaches were perfectly fit for the cases we studied but they were not able to manage the growing dimension, either because of stability issues for kernel-based approaches or the larger memory cost for MANDy\footnote{We were limitted for using MANDy by working with a 32GB RAM, which caused no limitation with m-tensor regression.}. \Cref{fig:convergence:1} shows the difference in growth rate of feature spaces for a polynomial kernel and the tensorized feature map. This demonstrates how tensor regression leverages overfitting to adapt for high dimensional problems with low amounts of data, hence the need for regularization to capture few important features. Both spectral truncation and subsampling lower the effective feature space dimensionality.

    \begin{figure}[ht!]
        \centering
        \begin{subfigure}[t]{.5\linewidth}
            \centerline{\includegraphics[width=\textwidth]{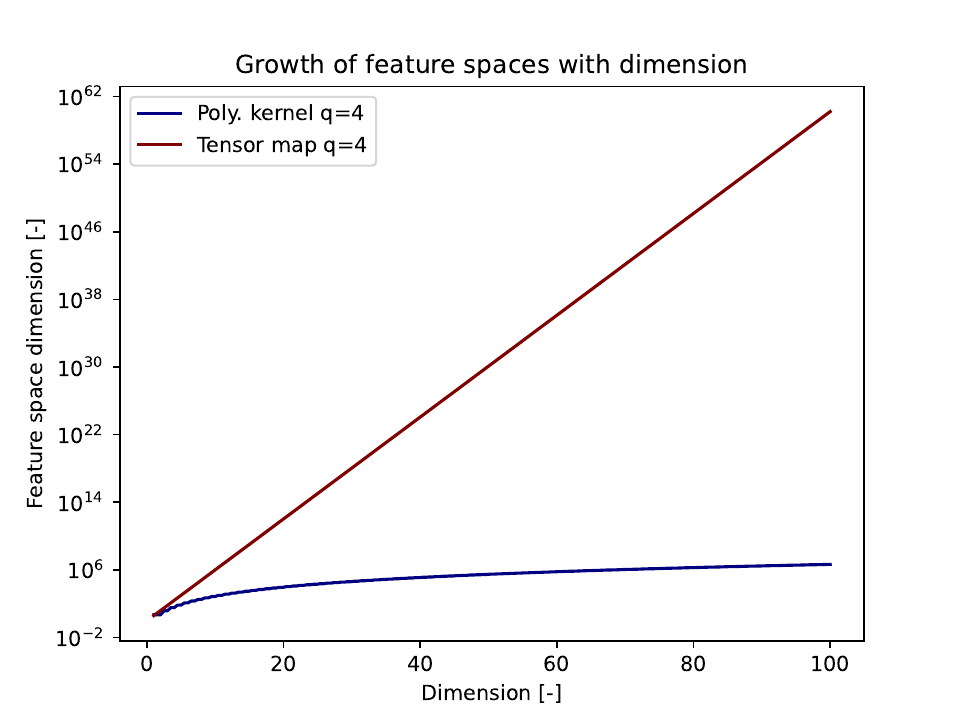}}
            \caption{Comparison of the growth of feature spaces for kernel polynomial regression and tensorized feature map.}
            \label{fig:convergence:1}
        \end{subfigure}%
        \begin{subfigure}[t]{.5\linewidth}
            \centerline{\includegraphics[width=\textwidth]{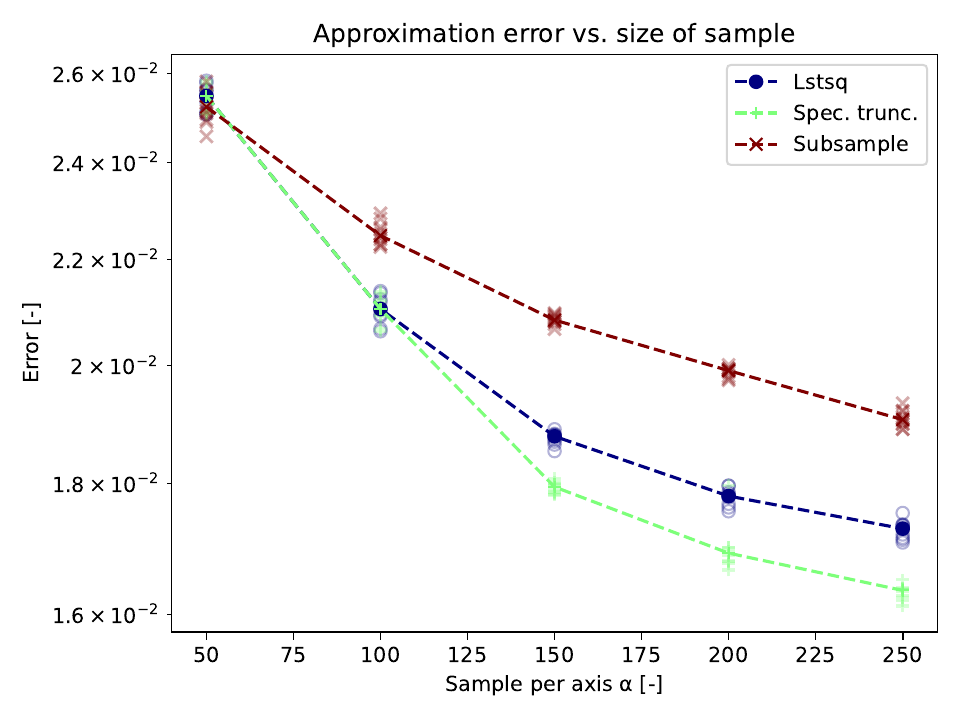}}
            \caption{Error of approximation on unseen data with $d=100$ for growing $\alpha$ (dashed line is the median).}
            \label{fig:convergence:2}
        \end{subfigure}
        \caption{Analysis of convergence of the model with the amount of data based on the tensorized feature map.}
        \label{fig:convergence}
    \end{figure}
    \Cref{fig:convergence:2} shows the error of the convergence study with growing $\alpha$. We can see the relatively slow convergence and the quality of regularization brought by spectral truncation. Subsampling remains around the same order of magnitude for approximately half the inference time. This demonstrates that even with the risk of overfitting, our approach effectively converges with good results at about 1\% of error in 100 dimension.


\section{Conclusion and perspectives}
\label{sec:conclu}

    In this paper we presented a tensor formalism for high dimensional regression of nonlinear models. The introduced m-tensor formalism is tailored for high dimensional regression in scarce data contexts. By leveraging tensor algebra, we provided properties from kernel methods with good scalability. The formalism smoothly generalizes well-known properties of matrix-based approaches to provide a tensor-structured least squares from a given mapping. We adapted different regularization techniques that aim at making models robust and generalizable from the low amount of data available. Those regularization techniques yield low errors while limiting the scaling when faced with a high number of parameters. The proposed subsampling approach allows for further gains in computational and memory costs.

    The m-tensor formalism was shown to avoid direct dependence on the exponential complexity of high dimensionality. Linear scaling of memory use with the dimension of the problem and parallel capability are inherited from the separated variables approach. We preserved an explicit feature map instead of an implicit one through a kernel function so as to maintain flexibility along each axis and scale mappings in a systematic way to guarantee stability. The resulting non-iterative approach is particularly effective in high dimensional contexts. Its strength lies in generating a feature space that grows rapidly with dimension while maintaining tractable computations through the proposed factorization.

    The proposed approach was applied on a general benchmark to demonstrate scaling and low relative errors. Behaviour with respect to the dimension, the order of approximation and amount of training input are analyzed. It should be noted that we still rely heavily on the sample size and do not offer vast guarantees for fully unknown models. The subsampling approach may however guide an adaptive sampling strategy. The proposed example holds for local approximation with known nonlinearity, we merely show how our framework preserves good approximation with minimal data in high dimension. Low dimensions can be seen as a weakness compared to state-of-the-art techniques as the mapping requires too much data to compare with oversampled approaches. 

    Perspectives include a study of general purpose approximations such as separable piecewise linear mappings, including support size considerations and minimal sample scaling. We will also seek different developments, most probably stemming from randomized numerical linear algebra to further accelerate our tools and enhance their scaling. On the application side, the tool seems most adapted to model order reduction, data-driven operator inference or classification. Large-scale optimization will be investigated along with the differenciation properties of the approach. The biggest challenge with such a large feature map is that of feature selection, \emph{i.e.} sparsification of the solution. An expansion of LASSO for $\ell_1$-norm minimization or probabilistic approaches may be of future interest.

\subsection*{Acknowledgements}We acknowledge fundings from the ENSAM RTE Chair, insightful discussions with Antonio Falcò and rich and fruitful exchanges with Jad Mounayer.


\bibliographystyle{siamplain}
\bibliography{biblio}

\newpage
\appendix

\section{Complements for stability analysis}
    \label{sec:appendix:stability}
    
    \subsection{Proofs}In this appendix, we provide proofs for the bounds provided in \cref{sec:numerical_stability} for the numerical stability of our approach. In that section we drop the $\ltimes$ for ease of reading. 

        \begin{proof}(\Cref{th:bound_phixdx})
            From the assumed structure, we have
            \begin{equation*}
                \bphi(\bx)-\bphi(\bx+\dbx)=\bigotimes_{i=1}^d\bpsi_i([\bx]_i)-\bigotimes_{i=1}^d
            \biggl(\bpsi_i([\bx]_i)+\pmb{\xi}_i([\bx]_i, [\dbx]_i)\biggr).
            \end{equation*}
            Elements of $\pmb{\xi}_{i=1}^d$ are assumed weighted by infinitesimal perturbations, then by neglecting all (at least) doubly infinitesimal terms in the product we have
            \begin{equation*}
                \bphi(\bx)-\bphi(\bx+\dbx)\approx\sum_{k=1}^d\biggl(\pmb{\xi}_k([\bx]_k, [\dbx]_k)\bigotimes_{\substack{i=1\\i\neq k}}^d\bpsi_i([\bx]_i)\biggr).
            \end{equation*}
            This leads to the bound in Frobenius norm
            \begin{align*}
                \left\lVert\bphi(\bx+\dbx)-\bphi(\bx)\right\rVert_F&\approx\left\lVert\sum_{k=1}^d\biggl(\pmb{\xi}_k([\bx]_k, [\dbx]_k)\bigotimes_{\substack{i=1\\i\neq k}}^d\bpsi_i([\bx]_i)\biggr)\right\rVert_F,\\
                &\lesssim\left\lVert\sum_{k=1}^d\biggl(\pmb{\xi}_k([\bx]_k, [\dbx]_k)\bigotimes_{i=1}^{d-1}\mathds{1}_p\biggr)\right\rVert_F,\\
                &\lesssim\sum_{k=1}^d\sqrt{\sum_{i_1=1}^p\sum_{i_2=1}^p\dots\sum_{i_d=1}^p[\pmb{\xi}_k([\bx]_k, [\dbx]_k)]_{i_1}^2},\\
                &\lesssim\sum_{k=1}^d\sqrt{\sum_{i_1=1}^p[\pmb{\xi}_k([\bx]_k, [\dbx]_k)]_{i_1}^2p^{d-1}},\\
                &\lesssim\sqrt{p^{d-1}}\sum_{k=1}^d\lVert\pmb{\xi}_k([\bx]_k, [\dbx]_k)\rVert_F.
            \end{align*}
            We assumed a worst \emph{normalized} case for the bound on $\pmb{\psi}_{i=1}^d$ so as to simplify our bound for interpretation.
        \end{proof}
        
        \begin{proof}(\Cref{lem:inputroundoff})
            The lemma stems from the bound of \cref{th:bound_phixdx} and assumes a constant infinitesimal perturbation $[\dbx]_{i=1}^d=\mO(\epsilon_\text{machine})$. Then we assumed the element-wise bound $[\pmb{\xi}_{i=1}^d([\bx]_i,[\dbx]_i)]_{j=1}^p\leq c\epsilon_\text{machine}$ which, plugged into the bound of \cref{th:bound_phixdx} leads to 
            \begin{equation*}
                \left\lVert\bphi(\bx+\dbx)-\bphi(\bx)\right\rVert_F\lesssim\sqrt{p^{d-1}}\sum_{k=1}^d\sqrt{\sum_{i=1}^pc^2\epsilon_\text{machine}^2}=dc\epsilon_\text{machine}\sqrt{p^d}.
            \end{equation*}
            The bounding constant $c_{\epsilon}$ is dependent on the choice of mapping and its domain of definition.
        \end{proof}

        \begin{proof}(\Cref{theorem:backward_stability})
            We assumed normalized mappings $\pmb{\psi}_{i=1}^d$ with separable perturbation-induced terms $\pmb{\xi}_{i=1}^d$ and settled that there exists a mapping-dependent constant $c$ such that $\lVert\pmb{\xi}_i(x,\delta x)\rVert_F\leq c\sqrt{p}$, which merely assumes $\pmb{\psi}_{i=1}^d$ remains bounded over its domain of definition. Then by denoting
            \begin{equation*}
                f(\bx)=\bphi(\bx)\bC,\quad\tilde{f}(\bx)=\bphi(\bx)(\bC+\dbC)\quad\text{and}\quad\delta\pmb{\phi}(\bx,\dbx)=\bphi(\bx)-\bphi(\bx+\dbx),
            \end{equation*}
            the backward error is given from its definition \cite{numericalAnalysis} as
            \begin{align*}
                \dfrac{\lVert\tilde{f}(\bx)-f(\bx+\dbx)\rVert}{\lVert f(\bx+\dbx)\rVert}&=\dfrac{\lVert\bphi(\bx)\bC+\bphi(\bx)\dbC-\bphi(\bx+\dbx)\bC\rVert}{\lVert\bphi(\bx+\dbx)\bC\rVert}\\
                &=\dfrac{\lVert \bphi(\bx)\bC+\bphi(\bx)\dbC-(\bphi(\bx)+\dbphi(\bx,\dbx))\bC\rVert}{\lVert(\bphi(\bx)+\dbphi(\bx,\dbx))\bC\rVert}\\
                &=\dfrac{\lVert \bphi(\bx)\dbC-\dbphi(\bx,\dbx)\bC\rVert}{\lVert(\bphi(\bx)+\dbphi(\bx,\dbx))\bC\rVert}\\
                &\leq\dfrac{\lVert\dbC\rVert}{\lVert\bC\rVert}\dfrac{\lVert\bphi(\bx)\rVert}{\lVert\bphi(\bx)+\dbphi(\bx,\dbx)\rVert}+\dfrac{\lVert\dbphi(\bx,\dbx)\rVert}{\lVert\bphi(\bx)+\dbphi(\bx,\dbx)\rVert}
            \end{align*}
            We note that while $\dbphi(\bx,\dbx)\ll\bphi(\bx)$, the criterion may be approximated as
            \begin{equation}
                \label{eq:intermediate_eq_bwd_stb}
                \dfrac{\lVert\dbC\rVert}{\lVert\bC\rVert}\dfrac{\lVert\bphi(\bx)\rVert}{\lVert\bphi(\bx)+\dbphi(\bx,\dbx)\rVert}+\dfrac{\lVert\dbphi(\bx,\dbx)\rVert}{\lVert\bphi(\bx)+\dbphi(\bx,\dbx)\rVert}\approx\dfrac{\lVert\dbC\rVert}{\lVert\bC\rVert}.
            \end{equation}
            From the bound provided in \cref{th:bound_phixdx} and considering the bounding constant $c$ introduced earlier, we have
            \begin{equation*}
                \lVert\dbphi(\bx,\dbx)\rVert\lesssim dc\sqrt{p^d}.
            \end{equation*}
            Assuming the \emph{normalized} bound $\pmb{\psi}_{i=1}^d\leq\mathds{1}_p$, \emph{i.e.} $\lVert\pmb{\varphi}(\bx)\rVert\leq\sqrt{p^d}$, the left-hand-side of \cref{eq:intermediate_eq_bwd_stb} is bounded as
            \begin{equation*}
                \dfrac{\lVert\dbC\rVert}{\lVert\bC\rVert}\dfrac{\lVert\bphi(\bx)\rVert}{\lVert\bphi(\bx)+\dbphi(\bx,\dbx)\rVert}+\dfrac{\lVert\dbphi(\bx,\dbx)\rVert}{\lVert\bphi(\bx)+\dbphi(\bx,\dbx)\rVert}\leq\dfrac{\lVert\dbC\rVert}{\lVert\bC\rVert}\dfrac{1}{1+dc}+\dfrac{dc}{1+dc}.
            \end{equation*}
            Thus, guaranteeing $dc\ll1$ lets the backward stability criterion tend only to the relative model error in $\bC$.
        \end{proof}

        \begin{proof}(\cref{th:bound_noise})
            We assume the sample $\mP$ to be sufficient and $\bphi$ to be well suited to approximate $f$, therefore 
            $$f(\bx)=\bphi(\bx)\bC=\bphi(\bx)\bPhi^\intercal(\bPhi\bPhi^\intercal)^{-1}\mathbf{y}.$$
            Then, the model error $\dbC$ with the introduction of noise behaves as
            \begin{equation*}
                \bPhi(\bC+\dbC)=\mathbf{y}+\delta\mathbf{y}.
            \end{equation*}
            So, following the first assumption, the model computation leads to
            \begin{equation*}
                \dbC=\mathbf{\Phi}^\intercal\left(\mathbf{\Phi}\mathbf{\Phi}^\intercal\right)^{-1}\delta\mathbf{y}.
            \end{equation*}
            Hence the bound
            \begin{equation*}
                \lVert\dbC\rVert\leq\dfrac{\lVert\delta\mathbf{y}\rVert}{\sigma_r},
            \end{equation*}
            since we leverage the mode-1 unfolding's pseudoinverse, whose norm is the reciprocal of the lowest nonzero singular value $\sigma_r$.
        \end{proof}

    \subsection{Stability for given mappings}
        We state in the paper that the choice of mapping implies the study of input perturbation. We provide that analysis for the generic canonical polynomial mapping and a trigonometric polynomial mapping.

        \subsubsection{Polynomial mapping}
            \label{sec:appendix:poly_stability}
            We consider perturbation $\delta x$ in the input to assess stability as defined in the case of the canonical 1D degree $q$ polynomial basis, so we consider
            \begin{equation}
                \label{eq:poly_basis}
                \bpsi(x+\delta x)=\begin{bmatrix}
                1 & x+\delta x & (x+\delta x)^2 & \dots & (x+\delta x)^q
                \end{bmatrix}.
            \end{equation}
            We immediately see how the choice of basis and normalization will influence stability. The $d$-dimensional tensorization of such bases will lead to the evaluation 
            \begin{equation}
                \label{eq:poly_sep_perturb}
                \pmb{\varphi}(\mathbf{x}+\delta\mathbf{x})=\bigotimes_{i=1}^d\bpsi([\bx]_i+[\dbx]_i).
            \end{equation}
            We manage to identify and separate the term due to input perturbation along each dimension by developping the basis \cref{eq:poly_basis} as 
            \begin{equation*}
                \bpsi(x+\delta x)=\begin{bmatrix}
                1 & x+\delta x & x^2+2x\delta x+\delta x^2 & \dots & \sum_{k=0}^q\binom{q}{k}x^k\delta x^{q-k}
                \end{bmatrix}.
            \end{equation*}
            In a first order approximation, we may separate that basis as $\bpsi(x+\delta x)\approx\bpsi(x)+\pmb{\xi}(x,\delta x)$ where
            \begin{equation}
                \label{eq:xi_poly}
                \pmb{\xi}(x,\delta x)=\begin{bmatrix}
                0 & \delta x & 2x\delta x & \dots & qx^{q-1}\delta x
                \end{bmatrix}.
            \end{equation}
            Therefore the multidimensional basis $\bphi$ inherits the partition
            \begin{equation}
                \label{eq:separated_xi}
                \bphi(\bx+\dbx)\approx\bphi(\bx)+\dbphi(\bx,\dbx)=\bphi(\bx)+\sum_{i=1}^d\biggl(\pmb{\xi}([\bx]_i,[\dbx]_i)\bigotimes_{\substack{j=1 \\ j\neq i}}^d\bpsi([\bx]_j)\biggr).
            \end{equation}
            We stated in \cref{theorem:backward_stability} that we need to identify a bounding constant $c$ from $\pmb{\xi}_{i=1}^d$ to assess the scaling to apply to remain stable. Let each $[\bx]_{i=1}^d$ be defined over an interval $[a_i,b_i]_{i=1}^d\subset\mathbb{R}$, then we can consider 
            $$\lVert\pmb{\xi}_i([\bx]_i,[\dbx]_i)\rVert_F<q\operatorname{max}(\lvert a_i\rvert,\lvert b_i\rvert)^{q-1}[\dbx]_i\sqrt{p}$$
            with quite a margin. Therefore, we may uniformly scale down $\bpsi_{i=1}^d$ with a factor $\mO(1/q\operatorname{max}(\lvert a_i\rvert,\lvert b_i\rvert)^{q-1})$. Note that for this mapping to be normalized element-wise, it is required that 
            $$\operatorname{max}(\lvert a_i\rvert,\lvert b_i\rvert)\leq1.$$
            The scale-down approach will not affect conditioning of the regression (which can be tackled with any proposed regularization technique) and guarantees stability as $dc\ll1$ seeing \cref{theorem:backward_stability}.

        \subsubsection{Trigonometric mapping}
            \label{sec:appendix:trigo_stability}
            For trigonometric approximation, we consider the 1D mapping
            \begin{equation}
                \label{eq:trigo_basis}
                \bpsi(x+\delta x)=\begin{bmatrix}
                1 & \operatorname{sin}(x+\delta x) & \operatorname{sin}^2(x+\delta x) & \dots & \operatorname{sin}^q(x+\delta x)
                \end{bmatrix}.
            \end{equation}
            This can be analogously developed for a basis in $\operatorname{cos}^{1\leq i\leq q}(x+\delta x)$. Those are naturally normalized element-wise. We start by noting that 
            \begin{equation*}
                \operatorname{sin}(x+\delta x)=\operatorname{sin}(x)-\operatorname{cos}(x)\operatorname{tan}(\delta x)\ \text{and}\ \operatorname{cos}(x+\delta x)=\operatorname{cos}(x)-\operatorname{sin}(x)\operatorname{tan}(\delta x).
            \end{equation*}
            Then by applying the polynomial case we can identify the structure of \cref{eq:separated_xi} with, for the polynomial sine basis,
            \begin{equation*}
                \label{eq:xi_trigo}
                \pmb{\xi}(x, \delta x)=\begin{bmatrix}
                    0 & -\operatorname{cos}(x)\operatorname{tan}(\delta x) & \dots & -q\operatorname{sin}^{q-1}(x)\operatorname{cos}(x)\operatorname{tan}(\delta x)
                \end{bmatrix}.
            \end{equation*}
            Chosing a loose bound from an element-wise worst case, we can provide
            $$\lVert\pmb{\xi}(x, \delta x)\rVert_F< \sqrt{q}\left(1-\dfrac{1}{q}\right)^{\frac{q-1}{2}}\operatorname{tan}(\delta x)\sqrt{p},$$
            so that scaling uniformly with a factor $\mO(q^{-1/2}(1-q^{-1})^{(1-q)/2})$ each 1D mapping will preserve the stability criterion of \cref{theorem:backward_stability}. This value comes from considering $[\pmb{\xi}_{i=1}^d(x,\delta x)]_{j=1}^p=\operatorname{max}_x\left(\lvert-q\operatorname{sin}^{q-1}(x)\operatorname{cos}(x)\operatorname{tan}(\delta x)\rvert\right)$ as a function of $q$.

\section{Algorithms to subsample the operator}
\label{sec:appendix:subsample}
    In this section we provide algorithms to compute subsampled m-tensors. We provide two ways of constructing it. The first is the streamed way, inspired from online learning of streamed data. We qualify the second as the greedy decomposition. We first review the streamed way and then provide arguments for the more costful greedy decomposition. We note that review of randomized numerical linear algebra and the work of Mahoney and Drineas on CUR matrix decomposition \cite{CUR-data} lead us to seek a \emph{near optimal} greedy algorithm leveraging randomized approaches with a low computational cost and preserving the efficiency of the streamed approach instead of the current greedy approach, closer to the pseudoskeleton approach of Goreinov et al. \cite{skeleton-CUR}. We settle notations for this section: we consider an m-tensor $\mathbf{T}$ with c-dim $m$ and provide its decomposition based on a subsample with a tolerance $\varepsilon$ as 
    $$\lVert\mathbf{T}-W\tilde{\mathbf{T}}\rVert^2\leq\varepsilon.$$
    In the remainder of the section, we give a geometric interpretation to the incremental construction of the subset of rows from a given m-tensor. We consider each row to be a point in a high dimensional space. Providing the distance definition from the inner product of two of those points.

    \subsection{Streamed decomposition}

        The streamed construction of the subset considers that we sweep once through the rows of an m-tensor $\mathbf{T}$ or treat rows as they are streamed to select the subset on-the-fly.

        \subsubsection*{Elements $1$ and $2$}
        The first provided row $[\mathbf{T}]_{1}$ is used by default to initialize $\tilde{\mathbf{T}}$. Weights $W_1$ are initialized as an array containing only 1. We define 
        $$L_1=\lVert[\mathbf{T}]_1\rVert.$$
        From that given point, we build a subspace $\operatorname{span}(\{[\mathbf{T}]_{1}\})$. For the next point, we evaluate the distance to that subspace. For a new row $[\mathbf{T}]_2$, we compute $\delta_2$ defined as
        \begin{equation}
            \label{eq:delta2}
            \delta_2=\underset{w\in\mathbb{R}}{\operatorname{min}}\left\lVert[\mathbf{T}]_2-w[\mathbf{T}]_1\right\rVert^2,
        \end{equation}
        where the minimizing $w$ is obtained by expanding the norm in the definition of $\delta$ and deriving the minimum of a quadratic polynomial in $w$, yielding
        $$w=\dfrac{\langle[\mathbf{T}]_1,\ [\mathbf{T}]_2\rangle}{\lVert[\mathbf{T}]_1\rVert^2}.$$
        Using that minimizer, $\delta$ will be given by
        $$\delta_2=\lVert[\mathbf{T}]_2\rVert^2-\left(\dfrac{\langle[\mathbf{T}]_2,\ [\mathbf{T}]_1\rangle}{\lVert[\mathbf{T}]_1\rVert}\right)^2.$$
        By comparing $\delta_2$ to the precision criterion $\varepsilon$, the point is kept or not. If it is kept, $[\mathbf{T}]_2$ is appended to $\tilde{\mathbf{T}}$ and $L$ is updated as
        $$L_2=\begin{bmatrix}
            L_{1} & 0\\
            s & \sqrt{\lVert[\mathbf{T}]_{2}\rVert^2-s^2}
        \end{bmatrix},\quad\text{with}\quad s=\dfrac{\langle[\mathbf{T}]_1,\ [\mathbf{T}]_2\rangle}{L_1}.$$
        The weight matrix $W$ is updated as
        $$W_2=\begin{bmatrix}
            W_1 & 0\\
            0 & 1
        \end{bmatrix}\quad\text{if }\delta_2\geq\varepsilon,\quad W_2=\begin{bmatrix}
            W_1 \\
            w
        \end{bmatrix}\quad\text{otherwise}.$$
        
        \subsubsection*{Element $k$}
        Consider we retained $\tilde{k}$ elements in $\tilde{\mathbf{T}}$ until that point. This means we need to evaluate the distance of $[\mathbf{T}]_k$ to a $\tilde{k}$-dimensional subspace. We define the vector $\mathbf{b}$ from the already selected rows and the new one as
        $$\mathbf{b}=\tilde{\mathbf{T}}\ltimes[\mathbf{T}]_k^\intercal.$$
        The squared distance $\delta_k$ between $[\mathbf{T}]_k$ and $\operatorname{span}(\tilde{\mathbf{T}})$ is given as
        \begin{equation}
            \label{eq:delta_k_streamed}
            \delta_k=\underset{\mathbf{w}}{\operatorname{min}}\left\lVert[\mathbf{T}]_k-\mathbf{w}^\intercal\tilde{\mathbf{T}}\right\rVert^2=\lVert[\mathbf{T}]_k\rVert^2-\mathbf{b}^\intercal\tilde{P}^{-1}\mathbf{b},
        \end{equation}
        where $\tilde{P}=\tilde{\mathbf{T}}\ltimes\tilde{\mathbf{T}}^\intercal$ admits a Cholesky decomposition with factor $L_{\tilde{k}}$. Therefore by computing the vector $\mathbf{s}=L_{\tilde{k}}^{-1}\mathbf{b}$ by downward backsubstitution, the squared distance $\delta_k$ can be efficiently computed from $\mathbf{b}^\intercal\tilde{P}^{-1}\mathbf{b}=\mathbf{s}^\intercal\mathbf{s}=\lVert\mathbf{s}\rVert^2$. If $\delta_k>\varepsilon$, $[\mathbf{T}]_k$ is appended to $\tilde{\mathbf{T}}$. The Cholesky factor is updated as
        $$L_{\tilde{k}+1}=\begin{bmatrix}
            L_{\tilde{k}} & 0\\
            \mathbf{s}^\intercal & \sqrt{\lVert[\mathbf{T}]_{k}\rVert^2-\lVert\mathbf{s}\rVert^2}
        \end{bmatrix}.$$
        The weight matrix $W$ is updated as
        $$W_{k+1}=\begin{bmatrix}
            W_k & 0\\
            0 & 1
        \end{bmatrix}\quad\text{if }\delta_k\geq\varepsilon,\quad W_{k+1}=\begin{bmatrix}
            W_k \\
            \mathbf{w}^\intercal
        \end{bmatrix}\quad\text{otherwise},$$
        with $\mathbf{w}=(L_{\tilde{k}}^\intercal)^{-1}\mathbf{s}$ computed by upward backsubstitution. This construction loops only once on the rows of $\mathbf{T}$ and leverages the incrementally constructed Cholesky factors $L_{\tilde{k}}$. It is fit for streamed data if the criterion on squared distance $\varepsilon$ can be chosen a priori. \Cref{alg:streamedALI} sums up the approach.

        \begin{algorithm}[hbt!]
            \caption{Streamed subsample construction}\label{alg:streamedALI}
            \begin{algorithmic}
                \Require{$\mathbf{T}$, $\varepsilon$}
                \Ensure{$\tilde{\mathbf{T}}$, $W$}
                \State $\tilde{\mathbf{T}}\gets[\mathbf{T}]_1$\hfill\Comment Initialize reduced m-tensor
                \State $L_{1}\gets\lVert[\mathbf{T}]_1\rVert$\hfill\Comment Initialize Cholesky factor
                \State $\tilde{k}\gets1$\hfill\Comment Initialize row index in reduced m-tensor
                \State $W\gets[1]$\hfill\Comment Initialize weights
                \For{any new $[\mathbf{T}]_k,\ k>1$}\hfill\Comment Loop on rows of initial $\mathbf{T}$ or streamed data
                \State $\mathbf{b}\gets\tilde{\mathbf{T}}\ltimes[\mathbf{T}]_k^\intercal$\hfill\Comment Compute $\mathbf{b}$
                \State $\mathbf{s}\gets L_{\tilde{k}}^{-1}\mathbf{b}$\hfill\Comment Compute $\mathbf{s}$ by backsubstitution
                \State $\delta_k\gets\lVert[\mathbf{T}]_k\rVert^2-\mathbf{s}^\intercal\mathbf{s}$\hfill\Comment Evaluate $\delta_k$
                \If{$\delta_k>\varepsilon$}\hfill\Comment Compare distance to threshold
                \State $\tilde{\mathbf{T}}\gets\begin{bmatrix}
                    \tilde{\mathbf{T}} & [\mathbf{T}]_k
                \end{bmatrix}^\intercal$\hfill\Comment Update reduced m-tensor $\tilde{\mathbf{T}}$
                \State $\tilde{k}\gets\tilde{k}+1$\hfill\Comment Update number of kept tensors $\tilde{k}$
                \State $L_{\tilde{k}}\gets\begin{bmatrix}
                    L_{\tilde{k}-1}&0\\
                    \mathbf{s}^\intercal&\sqrt{\lVert[\mathbf{T}]_k\rVert^2-\mathbf{s}^\intercal\mathbf{s}}
                \end{bmatrix}$\hfill\Comment Update Cholesky factor $L$
                \State $W\gets\begin{bmatrix}
                    W&0\\
                    0&1
                \end{bmatrix}$\hfill\Comment Update weights $W$
                \Else
                \State $W\gets\begin{bmatrix}
                    W & (L_{\tilde{k}}^\intercal)^{-1}\mathbf{s}
                \end{bmatrix}^\intercal$\hfill\Comment Update weights $W$
                \EndIf
                \EndFor
            \State\Return $\tilde{\mathbf{T}}$, $W$ 
            \end{algorithmic}
        \end{algorithm}

    \subsection{Greedy decomposition}

        From the same notations, we propose an iterative approach during which we add the best choice among all rows of $\mathbf{T}$ at each step, \emph{i.e.} the greedy version.

        \subsubsection*{Element $1$}
        The first row of $\tilde{\mathbf{T}}$ is selected so that its span minimizes the sum of distances with the other rows. To evaluate the distance of all points to the subspaces spanned by each points, we build the matrix $\Delta_1\in\mathbb{R}^{m\times m}$ as
        \begin{equation}
            \label{eq:Delta1}
            [\Delta_1]_{ij}=\delta_{ij}^1\quad,\quad\delta_{ij}^1=\underset{w\in\mathbb{R}}{\operatorname{min}}\left\lVert[\mathbf{T}]_i-w[\mathbf{T}]_j\right\rVert^2.
        \end{equation}
        Each element is defined analogously to the squared distance \cref{eq:delta2}. We compute the contraction of each column vector of $\Delta_1$ and select the first row of $\tilde{\mathbf{T}}$ so as to minimize that sum. If all distances in the column vector corresponding to the selected point were below the precision tolerance $\varepsilon$, then $\tilde{\mathbf{T}}$ would be sufficient and the algorithm would end there. Note that \emph{i}) the elements of $\Delta_1$ can be built from the elements of $P=\mathbf{T}\ltimes\mathbf{T}^\intercal$ and \emph{ii}) $\Delta_1$ is not symmetric.

        \subsubsection*{Element $k$}
        The $k$-th row of $\tilde{\mathbf{T}}$, $k>1$ is selected so that the subspace
        $$\operatorname{span}\left(\left\{[\tilde{\mathbf{T}}]_i,\ 1\leq i\leq k\right\}\right)$$
        minimizes the distances to all remaining points. In order to evaluate these distances, we denote $\bar{\mathbf{T}}_k$ the m-tensor $\mathbf{T}$ with rows previously selected in $\tilde{\mathbf{T}}$ masked. We define 
        $$\forall\ j,\ 1\leq j\leq m-k,\quad\tilde{\mathbf{T}}_j=\begin{bmatrix}
            \tilde{\mathbf{T}} & [\bar{\mathbf{T}}_k]_j
        \end{bmatrix}^\intercal$$
        as each possible augmentation of the subsample. At step $k$ we define the matrix $\Delta_k\in\mathbb{R}^{(m-k)\times(m-k)}$ element-wise as
        \begin{equation}
            \label{eq:Delta_k_greedy}
            \left[\Delta_k\right]_{ij}=\underset{\mathbf{w}}{\operatorname{min}}\lVert[\bar{\mathbf{T}}_i]-\mathbf{w}^\intercal\tilde{\mathbf{T}}_j\rVert^2.
        \end{equation}
        Each element is constructed analogously to \cref{eq:delta_k_streamed}. We contract the columns of $\Delta_k$ to obtain a vector $\pmb{\delta}_k$ whose elements are sums of squared distances. We select the index of the best choice in the rows of $\bar{\mathbf{T}}_k$ as the minimizer
        $$\hat{\imath}=\operatorname{argmin}_i([\pmb{\delta}_k]_i).$$
        The corresponding row $[\bar{\mathbf{T}}_k]_{\hat{\imath}}$ is appended to $\tilde{\mathbf{T}}$ and masked for following iterations. If all distances in the $\hat{\imath}$-th column of $\Delta_k$ are below $\varepsilon$, the precision criterion is reached and the algorithm is stopped. Otherwise, operators are updated and the loop continues.
    
        Once the optimal $\tilde{\mathbf{T}}$ of size $\tilde{m}$ is computed, we evaluate the weight matrix using the projector onto the span of $\tilde{\mathbf{T}}$ denoted by $\Pi_{\tilde{\mathbf{T}}}$ and defined as
        $$\Pi_{\tilde{\mathbf{T}}}=\tilde{\mathbf{T}}^\intercal\left(\tilde{\mathbf{T}}\ltimes\tilde{\mathbf{T}}^\intercal\right)^{-1}\tilde{\mathbf{T}}.$$
        We identify the weights in the projection of $\mathbf{T}$
        \begin{equation}
            \label{eq:greedy_coefs}
            \mathbf{T}\ltimes\Pi_{\tilde{\mathbf{T}}}=W\tilde{\mathbf{T}}.
        \end{equation}
        We provide \cref{alg:greedyALI} to build the greedy subsample from a full m-tensor.

        \begin{algorithm}[hbt!]
            \caption{Greedy subsample construction}
            \label{alg:greedyALI}
            \begin{algorithmic}
                \Require{$\mathbf{T}$, $\varepsilon$}
                \Ensure{$\tilde{\mathbf{T}}$, $W$}
                \State Build matrix $\Delta_1$ from \cref{eq:Delta1}
                \State $[\pmb{\delta}_1]_i\gets\sum_{j=1}^m[\Delta_1]_{ji}$\hfill\Comment Sum columns of $\Delta_1$
                \State $\tilde{\mathbf{T}}\gets[\mathbf{T}]_{\hat{\imath}},\ \hat{\imath}=\operatorname{argmin}_i([\pmb{\delta}_1]_i)$\hfill\Comment Select first element of $\tilde{\mathbf{T}}$
                \State $L_1\gets[\lVert\tilde{\mathbf{T}}\rVert]$\hfill\Comment Initialize Cholesky factor
                \State Define $\bar{\mathbf{T}}_1$ by masking $\hat{\imath}$-th row of $\mathbf{T}$
                \For{$k,\ 1< k\leq m$}\Comment Loop on rows of $\mathbf{T}$
                    \State Build matrix $\Delta_k$ from \cref{eq:Delta_k_greedy}
                    \State $[\pmb{\delta}_k]_i\gets\sum_{j=1}^{m-k}[\Delta_k]_{ji}$\hfill\Comment Sum columns of $\Delta_k$ in vector $\pmb{\delta}_k$
                    \State $\hat{\imath}=\operatorname{argmin}_i([\pmb{\delta}_k]_i)$\hfill\Comment Find optimal column index $\hat{\imath}$
                    \State $\tilde{\mathbf{T}}\gets\begin{bmatrix}
                        \tilde{\mathbf{T}} & [\bar{\mathbf{T}}]_{\hat{\imath}}
                    \end{bmatrix}^\intercal$\hfill\Comment Append $\hat{\imath}$-th row of $\bar{\mathbf{T}}$ as $k$-th row of $\tilde{\mathbf{T}}$
                    
                    \If{$\forall\ j,\ 1\leq j\leq m-k,\ [\Delta_k]_{j\hat{\imath}}\leq\varepsilon$}\Comment All dist. to current subspace under $\varepsilon$
                        \State Break loop
                    \Else
                    \State Update $\bar{\mathbf{T}}_{k+1}$ by masking $\hat{\imath}$-th row of $\bar{\mathbf{T}}_k$
                    \EndIf
                \EndFor
                \State Compute weights $W$ from \cref{eq:greedy_coefs}
                \State\Return $\tilde{\mathbf{T}}$, $W$
            \end{algorithmic}
        \end{algorithm}

\end{document}